\title{Bayesian Persuasion with Selective Disclosure\footnote{We thank Pak Hung Au,
Jeffrey Ely, Ying Gao, 
 Alexander Jakobsen,
Navin Kartik, Xiao Lin,
Elliot Lipnowski, Stephen Morris, Paula Onuchic,
Alessandro Pavan, Refine.ink, 
Jean Tirole, Alexander Wolitzky, Kun Zhang, and Gabriel Ziegler for helpful conversations and NSF grants SES-2337566 and
SES-2417162 for financial support.}}
\author{Yifan Dai\footnote{Department of Economics, Massachusetts Institute of Technology. Email: yfdai@mit.edu} \and
Drew Fudenberg\footnote{Department of Economics, Massachusetts Institute of Technology. Email: drew.fudenberg@gmail.com} \and
Harry Pei\footnote{Department of Economics, Northwestern University. Email: harrydp@northwestern.edu}}
\date{\today}
\documentclass[12pt]{article}

\usepackage{amsmath}
\usepackage{graphicx}
\usepackage{amsfonts}
\usepackage{amsthm}
\usepackage{setspace}
\usepackage{tikz}
\usepackage{amssymb}
\usepackage{appendix}
\usepackage{mathtools}   
\usepackage{microtype}
\usetikzlibrary{patterns}
\usepackage{sgame}
\usepackage{color}
\usepackage{soul}
\usepackage{enumitem}
\usepackage{multirow,array}
\usepackage[top=1in, bottom=1in, left=1in, right=1in]{geometry}

\usepackage[authoryear,round]{natbib}
\usepackage{hyperref}
\hypersetup{
	colorlinks=true,
	linkcolor=red!50!black,
	citecolor=blue!65!black,
	filecolor=magenta,      
	urlcolor=cyan,
}

\usepackage[nameinlink,noabbrev,sort,capitalise]{cleveref}

\DeclareMathOperator{\supp}{\textrm{supp}}

\begin{document}

\maketitle

\noindent \textbf{Abstract:} 
A sender first publicly commits to an experiment and then 
can  privately run additional experiments and selectively disclose their outcomes to a receiver. The sender has private information about the maximal number of additional experiments they can perform (i.e., their \textit{type}). 
 We show that the sender cannot attain their commitment payoff in any equilibrium
 if (i) the receiver is sufficiently uncertain about their type and (ii) the sender could benefit from selective disclosure after conducting their full-commitment optimal experiment.
 Otherwise, there can be  equilibria where the sender obtains their commitment payoff.\\
 


\newtheorem{Proposition}{\hskip\parindent\bf{Proposition}}
\newtheorem{Assumption}{\hskip\parindent\bf{Assumption}}
\newtheorem{Theorem}{\hskip\parindent\bf{Theorem}}
\newtheorem{Lemma}{\hskip\parindent\bf{Lemma}}
\newtheorem{Corollary}{\hskip\parindent\bf{Corollary}}
\newtheorem{Condition}{\hskip\parindent\bf{Condition}}
\newtheorem{Definition}{\hskip\parindent\bf{Definition}}
\newtheorem{Claim}{\hskip\parindent\bf{Claim}}
\newtheorem*{Condition1}{\hskip\parindent\bf{Condition FD}}
\thispagestyle{empty}

\newpage \setcounter{page}{1}

\begin{spacing}{1.5}
\section{Introduction}
The credibility of experts depends on their ability to commit to how they acquire and disclose information. For instance, a prosecutor who primarily seeks convictions but lacks commitment power cannot credibly convey useful information to a judge. But when the prosecutor can commit to an information-gathering process and to fully disclosing its results, as in the model of \citet{kamenica2011bayesian}, their statements can be informative and may influence the judge’s decision.

In practice, even when experts can commit to acquiring and truthfully disclosing certain types of information, they may be unable to  commit not to conduct additional investigations and not to conceal unfavorable findings.
This issue arises in the pharmaceutical industry, where some clinical trials are pre-registered with the FDA and must be reported, but  companies can conduct additional trials after approval. Because these post-approval trials are not pre-registered, companies sometimes conceal or delay the publication of their findings, even though they are legally required to report the results of all trials.\footnote{\citet*{devito2020compliance} reported that, among 4,209 trials obligated to submit results to the FDA between 2018 and 2019, only 722 (40.9\%) complied within the mandated 1-year time frame.
The pain medication Vioxx is a  high-profile example of delayed reporting: Its manufacturer, Merck, delayed releasing data from post-approval trials that linked the drug to heart problems. Another example is GSK’s antidepressant Paxil—legal filings revealed that the company concealed evidence from clinical trials showing that the drug could induce suicidal thoughts among teenagers. Compared with hiding or delaying disclosure, outright falsification of data is less common and typically subject to harsher penalties.} This can undermine the expert's credibility, as decision-makers (e.g., doctors, patients) may suspect that unfavorable evidence has been withheld and therefore discount the expert's statements.

Motivated by this concern, we study a model where a sender (the expert) publicly commits to an \textit{initial experiment} and discloses its outcome to a receiver. Unlike in the full commitment benchmark of \citet{kamenica2011bayesian}, the sender in our model cannot commit not to privately conduct additional experiments and conceal unfavorable results.  The sender has private information about the maximum number of additional experiments they can run, which is referred to as their \textit{type} or \textit{capacity}.

Our question is when the sender's limited commitment prevents them from attaining their payoff under full commitment (i.e., their \textit{commitment payoff}). Our results show that the answer depends on two factors: (i) the receiver's prior uncertainty about the sender's type and (ii) 
whether the sender benefits from selectively disclosing more information after conducting 
their full-commitment optimal  experiment  (hereafter, \textit{the optimal experiment}).

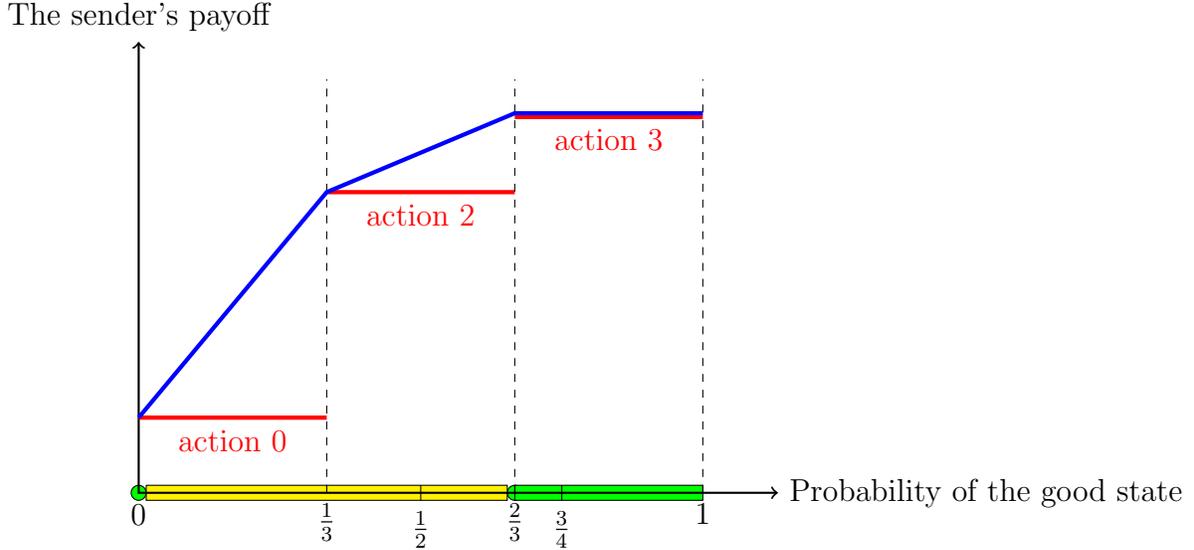
\begin{figure}
\begin{center}
\begin{tikzpicture}[scale=0.5]
\draw[fill=green] (0,0) circle [radius=0.2];
\draw[fill=green] (10,0) circle [radius=0.2];
\draw[fill=green] (10,0.2)--(15,0.2)--(15,-0.2)--(10,-0.2)--(10,0.2);
\draw[fill=yellow] (0.2,0.2)--(9.8,0.2)--(9.8,-0.2)--(0.2,-0.2)--(0.2,0.2);
\draw[thick, <->] (0,12)node[above]{The sender's payoff}--(0,0)node[below]{$0$}--(17,0)node[right]{Probability of the good state};
\draw[ultra thick, red] (0,2)--(2.5,2)node[below]{action $0$}--(5,2);
\draw[ultra thick, red] (5,8)--(7.5,8)node[below]{action $2$}--(10,8);
\draw[ultra thick, red] (10,10)--(12.5,10)node[below]{action $3$}--(15,10);
\draw[dashed] (5,0)node[below]{$\frac{1}{3}$}--(5,11);
\draw[dashed] (10,0)node[below]{$\frac{2}{3}$}--(10,11);
\draw[dashed] (15,0)node[below]{$1$}--(15,11);
\draw[ultra thick, blue] (0,2)--(5,8)--(10,10.1)--(15,10.1);
\draw (7.5,0.2)--(7.5,-0.2)node[below]{$\frac{1}{2}$};
\draw (11.25,0.2)--(11.25,-0.2)node[below]{$\frac{3}{4}$};
\end{tikzpicture}
\caption{The sender's utility under the receiver's best reply as a function of the receiver's belief (\textcolor{red}{red}) and its concave closure (\textcolor{blue}{blue}), with non-credible beliefs in yellow and credible beliefs in green.}
\end{center}
\end{figure}

To illustrate our findings, consider the  example  depicted in Figure 1: There are two states, \textit{good} and \textit{bad}, and three actions $0$, $2$, and $3$. The sender's utility equals the receiver's action. The receiver's optimal action is $0$ when the probability of the good state is less than $\frac{1}{3}$,  $3$ when the probability of the good state is more than $\frac{2}{3}$, and  $2$ otherwise. 

We first show that the sender's limited commitment matters only when the receiver faces uncertainty about the sender's capacity. This is because when the receiver knows that 
the sender can conduct at most $t$ additional experiments, the sender can prove that no information is hidden by  conducting $t$ uninformative additional experiments and disclosing all the outcomes.

When the receiver does not know the sender's capacity, the sender's ability to attain their commitment payoff 
hinges on the beliefs induced by the optimal experiment, in particular, whether the sender can benefit from selectively disclosing more information at any of those beliefs. In the example, the sender can benefit from disclosing more information at a belief that assigns probability $\frac{1}{2}$ to the good state: They can conduct an additional experiment that reveals the state and disclose only the outcome which proves that the state is good. This raises the receiver’s action from 
$2$ to $3$
when the state is good and keeps it at 
$2$ when the state is bad. In contrast, the sender cannot benefit from disclosing more information at beliefs that assign probability $0$ or $\frac{3}{4}$ to the good state  since the receiver's action will always be $0$ in the first case and will never exceed $3$ in the second case.
We refer to beliefs where the sender could gain by revealing more information as  \textit{non-credible beliefs}. In this example, these are the beliefs 
that assign probability strictly between $0$ and $\frac{2}{3}$ to the good state.

Theorem \ref{Theorem1} shows that in \textit{monotone environments} such as our example, the sender can attain their commitment payoff under all type distributions if and only if the optimal experiment does not induce a non-credible belief. If the optimal experiment induces at least one non-credible belief, then there is an open set of type distributions
under which all of the sender's equilibrium payoffs are uniformly bounded below their commitment payoff.

Theorem \ref{Theorem3} provides a complementary result that does not require monotonicity. It focuses on environments where the sender strictly prefers to fully disclose a certain state $\theta$
to any belief induced by the optimal experiment when the state is $\theta$.
This is more demanding than the optimal experiment inducing a non-credible belief, which is satisfied in our example
when the good state occurs with probability below $1/3$. We show that if sufficiently many sender types occur with probability bounded above $0$, the sender’s payoff in every equilibrium is uniformly bounded below their commitment payoff.

To show that the sender cannot attain their commitment payoff, we bound a higher type's equilibrium payoff from below by examining their payoff from the following deviation: first use a lower type's equilibrium strategy, then conduct a fully informative additional experiment and disclose its outcome if and only if doing so increases the sender's payoff. If the optimal experiment induces beliefs that satisfy our conditions, then in every equilibrium where the sender attains their commitment payoff, the higher type's payoff from our deviation exceeds the lower type's equilibrium payoff by a constant. 
Since payoffs are bounded, a contradiction arises once there are sufficiently many types that occur with positive probability.

We also show via an example that when the receiver's belief assigns positive probability only to a few sender types, the sender may obtain their commitment payoff even when the optimal experiment induces some non-credible belief.
We then show in 
Theorem \ref{Theorem2} that when one sender type occurs with probability close to $1$, there is an equilibrium where the sender's payoff is close to their commitment payoff. 

Our paper contributes to the literature on Bayesian persuasion by considering a novel form of limited commitment. While much of the literature focuses on the sender's inability to commit to a public experiment, we study the sender's inability to commit not to conduct additional private experiments. This generates a new economic friction: the sender's credibility is undermined by the receiver's suspicion about undisclosed follow-up investigations.

Thus our work differs from models where the sender can commit to a public experiment, but cannot commit to truthfully reveal what they learn (\citet*{lipnowski2022persuasion}), \citet{kreutzkamp2025persuasion}, \citet{lyu2025information}), or cannot commit to public experiments but can commit not to falsify information (\citet{henry2019research}, \cite{titova2025persuasion}), or cannot commit to public experiments but can communicate through a mediator with commitment (\citet*{myerson1982optimal}, \citet*{forges1986approach}, \citet*{goltsman2009mediation}, \citet{corrao2025bounds}) or cannot commit not to deviate to strategies that induce the same message distribution (\citet{lin2024credible}), or cannot commit in the stage game but interacts repeatedly with a sequence of receivers (e.g., \citet*{pei2023repeated}, \citet{best2024persuasion}, \citet*{mathevet2024reputation}).\footnote{\citet{lipnowski2020cheap} and \citet{kamenica2024commitment}  compare the sender's commitment payoff with their  payoff in a cheap talk game where they observe the state.}

In \citet{felgenhauer2014strategic} and  \citet{felgenhauer2017bayesian}, the sender cannot commit to a public experiment, but  can privately conduct a sequence of binary experiments and then select a subset of the resulting outcomes which they disclose to a receiver.\footnote{\citet{lou2023private} is similar to \citet{felgenhauer2017bayesian} except that once a sender discloses the outcome of an experiment, they must disclose the outcomes of all experiments conducted before that.} A more closely related paper  is  recent  independent work by \citet{st2025}. As in  our model, the receiver is  uncertain about how many additional experiments the sender can conduct. The paper  considers more general payoff functions for the sender  and provide conditions under which the commitment payoff is attained in at least one weak Perfect Bayesian Equilibrium (wPBE), a concept that is weaker than Perfect Bayesian Equilibrium (PBE). In contrast, we use an extension of PBE that reduces to it in finite signaling games.
We elaborate  on the effect of the  different  equilibrium concepts in Section \ref{sec4}.

In our model, the sender receives private information about their type before conducting additional experiments. This is related to the literature on information design with an informed sender, which includes the works of \citet{perez2014interim}, \citet{koessler2023informed}, and  \citet{zapechelnyuk2023equivalence}. In those papers, the sender's private information directly affects the receiver's payoff, while in our model it is payoff-irrelevant for the 
receiver.\footnote{\citet{rappoport2017incentivizing} study a model where the sender has private information about the single experiment that was conducted, but it  ends up not mattering for the receiver's belief about the state.}

Our paper is also related to the literature following \citet{dye1985disclosure} on the disclosure of evidence when the amount of evidence available to the sender is unknown to the receiver, and especially to \citet{dziuda2011strategic} and \citet{gao2025inference}  where a sender has multiple pieces of evidence and decides which subset of them to disclose. In our paper, the evidence available to the  sender is determined by their  choice of additional experiments. \citet{pei2025community} studies a repeated game with anonymous random matching, where players can selectively disclose signals generated by their past actions and players do not know how many signals their partners have generated.  This feature is similar to our model where the sender's limited commitment lowers their payoff only when the receiver faces uncertainty regarding the 
sender's capacity.

\section{The Baseline Model}\label{sec2}
A sender ($s$) interacts with a receiver ($r$). The sender's payoff $u^s(a)$ depends only on the receiver's action $a \in A$, so they have \emph{transparent motives} in the sense of \cite{lipnowski2020cheap}. The receiver's payoff $u^r(\theta,a)$ depends on both $a$ and the state of the world $\theta \in \Theta$. The state is initially unknown to both players, and is drawn according to a distribution $\pi_0 \in \Delta (\Theta)$ that has full support.
We assume that $\Theta$ and $A$ are finite sets.

An \textit{experiment} $\sigma$ is a mapping from $\Theta$ to the set of probability distributions over outcomes $s \in S$, where $S$ is a countably infinite set. Let $\Sigma$ denote the collection of all experiments, and for each experiment $\sigma \in \Sigma$, let $\sigma(s \mid \theta)$ denote the probability of $s$ conditional on $\theta$.

At the beginning of the game, the sender chooses an \textit{initial experiment}  $\sigma_0$, and then  both players observe $\sigma_0$ and its realized outcome $s_0 \in S$. After observing $(\sigma_0,s_0)$ but before the receiver chooses their action, the sender may conduct some \textit{additional experiments}. 
The outcomes of different experiments are assumed to be independent conditional on $\theta$. The maximal number of additional experiments that the sender can conduct is $t \in \{0,1,...\} \equiv \mathbb{N}$, which is  drawn according to $p \in \Delta (\mathbb{N})$. The sender privately observes $t$ after choosing $\sigma_0$.\footnote{Section \ref{sec4} discusses the case where the sender observes $t$ before choosing their initial experiment.} We refer to $t$ as the sender's (interim) \textit{type}.

The sender conducts additional experiments sequentially after observing $(\sigma_0,s_0,t)$, so their choice of the $k$th additional experiment can depend on $(\sigma_0,s_0,t)$ as well as the first $k-1$ additional experiments and their outcomes $\{\sigma_i,s_i\}_{i=1}^{k-1}$. The sender then chooses  a subset of  outcomes of the additional experiments to disclose before the receiver
chooses $a \in A$.

We assume that the receiver observes an additional experiment and its outcome if and only if the sender discloses that outcome. Hence, the receiver cannot directly observe $t$, the number of additional experiments the sender conducted, and the contents and outcomes of the undisclosed additional experiments. 


Let $\mathcal{H}_k$ denote the set of sequences that consist of $k$ pairs of an experiment $\sigma$ together with one of its outcomes $s$. 
Let $\mathcal{H} \equiv \bigcup_{k=1}^{+\infty} \mathcal{H}_k$ with a typical element denoted by $h \in \mathcal{H}$. The receiver's information sets are elements of  $\mathcal{H}$ corresponding to the disclosed experiments.\footnote{Whether the receiver observes the order of the disclosed experiments does not affect any of our results. }   
The receiver's strategy $\alpha: \mathcal{H} \rightarrow \Delta (A)$ maps their information sets to their actions.

The sender's information sets other than the initial node belong to $\mathcal{H} \times \mathbb{N}$.
Their pure strategy consists of (i) an initial experiment $\sigma_0 \in \Sigma$, 
(ii) a mapping from every type $t \geq 1$ and every $h\in \mathcal{H}_k$ with $1\leq k \leq t$ to $\Sigma \cup \{\emptyset\}$ that captures the sender's choice of the next experiment given their type and the outcomes of the previous experiments, where $\emptyset$ stands for conducting no further experiment,\footnote{The type $t$ sender can only choose $\emptyset$ (i.e., conduct no further experiment) for any  $h \in \mathcal{H}_{t+1}$.} and (iii) a mapping from every type $t \geq 1$ and every
 $h \in \mathcal{H}_{k}$ with $ 1 \leq k \leq t+1$ to subsets of $\{1,2,...,k-1\}$, 
where $i$ belongs to the chosen subset when the $i$th additional experiment is disclosed.
The sender's \textit{mixed strategy} is a distribution of their pure strategies.

The receiver's \textit{naive belief} at $h \in \mathcal{H}$, denoted by $\widehat{\pi}_h \in \Delta (\Theta)$, is defined as their belief  about $\theta$ if  they observe $h$ and believe that the sender has disclosed the outcomes of all the experiments they conducted. By definition, the receiver's naive belief depends only on the experiment-outcome pairs in $h$ but not on the order of these pairs and the sender's strategy.
The receiver's naive belief need not coincide with their posterior belief since they may suspect that the sender has concealed some outcomes from the additional experiments.

Our game is neither finite nor a multistage game with observable actions, so neither sequential equilibrium (\citet{kreps1982sequential}) nor Perfect Bayesian Equilibrium  (\citet{fudenberg1991perfect}) directly applies.  Instead, we will use \citet{fudenberg1991perfect}'s  \textit{Perfect Extended Bayesian Equilibrium} (PEBE, or sometimes, equilibrium) as our solution concept.

Formally,
 a PEBE for our game consists of a strategy for each player and a mapping from $\mathcal{H}$ to $\Delta (\Theta \times \mathbb{N})$ that captures the receiver's belief about $(\theta,t)$, such that (i) at every information set, each player's strategy best replies to their opponent's strategy given their belief, (ii) the receiver's belief is derived from Bayes rule on the equilibrium path, and is consistent with some conditional probability system induced by the sender's strategy at every information set.\footnote{\label{fn7}A conditional probability system (CPS) on set $\Omega$ is a function $\mu(\cdot|\cdot): 2^{\Omega} \times 2^{\Omega}\backslash \{\emptyset\} \rightarrow [0,1]$ such that (i) for all non-empty $\Omega_0 \subseteq \Omega$, $\mu(\cdot|\Omega_0)$ is a probability distribution on $\Omega_0$ and (ii) for all $\Omega_3 \subseteq \Omega_2 \subseteq \Omega_1 \subseteq \Omega$ with $\Omega_2 \neq \emptyset$, we have $\mu(\Omega_3|\Omega_2) \mu(\Omega_2|\Omega_1) =\mu(\Omega_3|\Omega_1)$. In our game, $\Omega$ is the set of feasible terminal nodes, so every CPS assigns conditional probability $0$ to the sender conducting more experiments than their type allows. Appendix A develops the very minor extension of the PEBE definition that is needed to handle the fact that the sender's choice of experiments induces a distribution of verifiable signals.  It shows that PEBE has the intuitively-appealing property that when $t$ is bounded above by $t^*$, the receiver's belief conditional on observing any $h \in \mathcal{H}_{t^*+1}$  equals their naive belief. This property is not required by the weak Perfect Bayesian Equilibrium, the solution concept used by \citet{st2025}.} 

\section{Analysis}\label{sec3}
We examine when the sender's limited commitment prevents them from attaining their \textit{commitment payoff}, i.e., their payoff in the benchmark of \citet{kamenica2011bayesian}. Our results suggest that the answer hinges on two factors 
(i) whether the receiver faces enough uncertainty about the sender's type and
(ii) whether the sender can benefit from selectively disclosing more information after conducting their optimal experiment.

\subsection{Assumptions \& Preliminaries}\label{sub3.1}
We now introduce Assumptions \ref{Ass1} and \ref{Ass2}, which will be maintained throughout the paper.
Let $A^*(\pi) \equiv \arg\max_{a \in A} \sum_{\theta \in \Theta} \pi(\theta) u^r(\theta,a)$ denote the set of receiver-optimal actions under $\pi \in \Delta (\Theta)$. 
\begin{Assumption}\label{Ass1}
For every $a \in A$ and $\pi \in \Delta (\Theta)$ such that $a \in A^*(\pi)$, 
there exists $\pi' \in \Delta(\Theta)$ with $\supp (\pi') \subseteq \supp (\pi)$ that satisfies
$A^*(\pi') = \{a\}$.
 \end{Assumption}
Assumption \ref{Ass1}  requires that for every action $a$ that is optimal for the receiver under some belief $\pi$ about the state, there exists $\pi'$ with a weakly smaller support under which $a$ is strictly optimal for the receiver. It implies that the receiver has a strict best reply when they know the state. Since $\Theta$ and $A$ are finite, Assumption \ref{Ass1} is satisfied for a set of $u^r(\theta,a)$ of full Lebesgue measure in $\mathbb{R}^{|\Theta| \times |A|}$ by Proposition 2 of \citet*{lipnowski2024perfect}.

Let $u(\theta)$ denote the sender's payoff when the receiver plays their (unique) best reply in state $\theta$. 
The sender's \textit{full-disclosure payoff} is
$\underline{V} (\pi_0) \equiv \sum_{\theta \in \Theta} \pi_0 (\theta) u(\theta)$. 

As in \citet{kamenica2011bayesian}, once the receiver's prior $\pi_0 \in \Delta (\Theta)$ is fixed, each experiment $\sigma$ is characterized by the distribution it induces over naive beliefs $\pi \in \Delta(\Theta)$, which we will refer to as \textit{beliefs} for short.   We use $\sigma[\pi]$ to denote the probability with which belief $\pi$ occurs 
in the distribution over beliefs that characterizes $\sigma$, and 
say that an experiment $\sigma$ \textit{induces} belief $\pi$ if $\sigma[\pi]>0$. 
We use $\Sigma[\pi] \subseteq \Delta (\Delta (\Theta))$ to denote the set of distributions over beliefs where the mean of those distributions equals $\pi$.
Let $\overline{u}(\pi) \equiv \max_{a \in  A^{*}(\pi)} u^s(a)$ 
denote the sender's highest payoff when the receiver's belief is $\pi$.
The sender's \textit{commitment payoff} is defined as
\begin{equation}\label{optimalexperiment}
\overline{V}(\pi_0) \equiv \max_{\sigma \in \Sigma [\pi_0]} \sum_{\pi \in \Delta(\Theta)} \sigma[\pi] \overline{u}(\pi).
\end{equation} 
By definition, $\underline{V}(\pi_0) \leq \overline{V}(\pi_0)$.
An  experiment is \textit{optimal} if it attains the maximum in (\ref{optimalexperiment}).  
Assumption \ref{Ass1} guarantees that each optimal experiment is the limit of a sequence of experiments that always induce strict receiver-incentives (see Lemma \ref{L1} in Appendix \ref{App: Assumptions}). A standard upper hemi-continuity argument then shows that 
when the sender can commit not to conduct additional experiments,
they can secure payoffs arbitrarily close to $\overline{V}(\pi_0)$, regardless of how the receiver breaks ties.\footnote{Hence, Assumption \ref{Ass1} leads  to the ``continuity property'' in \citet*{ali2024design}.}
To simplify our analysis, we focus on $(\pi_{0},u^s,u^r)$ under which there is a unique optimal experiment:
\begin{Assumption}\label{Ass2}
There is a unique optimal experiment, which is denoted by $\sigma^*$. 
\end{Assumption}
When $|\Theta|=2$, Assumption \ref{Ass2} is satisfied for all $(\pi_{0},u^s,u^r)$ except for those where 
$\overline{V}(\pi_0) = \max_{\pi \in \Delta(\Theta) }\overline{u}(\pi)$ and a zero-measure set of $(u^s,u^r) \in \mathbb{R}^{|\Theta| \times |A|^2}$.

Proposition \ref{Prop1} shows that the sender attains their commitment payoff in all equilibria unless 
(i) their commitment payoff exceeds their full-disclosure payoff, and (ii) the receiver faces nontrivial uncertainty about the sender’s capacity to conduct additional experiments.
\begin{Proposition}\label{Prop1}\leavevmode
\begin{enumerate}
    \item For any
    $p \in \Delta (\mathbb{N})$,
    the sender's payoff belongs to $[\underline{V}(\pi_0),\overline{V}(\pi_0)]$ in every PEBE.
\item  For any $p \in \Delta (\mathbb{N})$ that is degenerate,
the sender's payoff is $\overline{V}(\pi_0)$ in every PEBE.
\end{enumerate}
\end{Proposition}
The proof is in Appendix \ref{secAA}. Intuitively, the sender can secure $\underline{V}(\pi_0)$ by fully disclosing the state in the initial experiment. When the receiver knows $t$, the sender can secure payoff arbitrarily close to $\overline{V}(\pi_0)$ by (i) choosing an initial experiment that is close to $\sigma^*$ under which the receiver has strict incentives at each induced belief, which exists under Assumption \ref{Ass1},  and (ii) conducting $t$ non-informative additional experiments and disclosing their outcomes.

\subsection{Main Results}\label{sub3.2}
Our main results examine when the sender's limited commitment prevents them from attaining their commitment payoff in \textit{all} equilibria. 
We first introduce the definition of \textit{credible beliefs}:
\begin{Definition}
A belief $\pi \in \Delta (\Theta)$ about the state is \textup{credible} if  
for every $\theta \in \supp(\pi)$, we have $\overline{u}(\pi) \geq u(\theta)$.
Otherwise, this belief
is \textup{non-credible}.  
\end{Definition}
Both degenerate beliefs and beliefs where the sender attains their maximum utility are credible as illustrated in our initial example. Theorem \ref{Theorem1} shows that in \textit{monotone} payoff environments, which have been a primary focus of the persuasion literature, the sender can attain their commitment payoff under all type distributions \textit{if and only if} the optimal experiment $\sigma^*$ does not induce a non-credible belief. 
\begin{Definition} 
$(u^s,u^r)$ is \textup{monotone} if there exist complete orders  on $\Theta$ and $A$, respectively, under which $u^s(a)$ is strictly increasing in $a$ and $u^r(\theta,a)$ has strictly increasing differences. 
\end{Definition}
\begin{Theorem}\label{Theorem1}
Suppose $(u^s,u^r)$ is monotone. 
\begin{enumerate}
  \item If all beliefs induced by $\sigma^*$ are credible, then for every 
  $p \in \Delta (\mathbb{N})$,
  there exists a PEBE in which the sender's payoff is $\overline{V}(\pi_0)$.
  \item If $\sigma^*$ induces at least one non-credible belief, then there exist $\eta>0$ and an open set of type distributions (in the product topology) such that, under each type distribution $p$ in this set, 
the sender's payoff in every PEBE is no more than $\overline{V}(\pi_0)-\eta$. 
\end{enumerate}
\end{Theorem}
In the example of  the introduction, 
 the optimal experiment induces a non-credible belief when the prior probability of the good state lies in $(0, 2/3)$. Consequently, by Theorem \ref{Theorem1}, the sender’s payoff is strictly  below the commitment payoff for an open set of type distributions. By contrast, in the prosecutor–judge example of \cite{kamenica2011bayesian}, all beliefs induced by the optimal experiment are credible. As a result, the sender can attain the commitment payoff in some equilibrium, irrespective of the type distribution.

To show the first statement, suppose the sender chooses $\sigma^*$ as the initial experiment and conducts no additional experiments on the equilibrium path; the receiver breaks ties in favor of the sender; and the receiver's posterior belief at every information set equals their naive belief. This is an equilibrium in monotone environments since $\overline{u}(\cdot)$ is increasing in first-order stochastic dominance order: 
If $\pi \in \Delta (\Theta)$ satisfies $\overline{u}(\pi) \geq u(\theta)$ for every $\theta \in \supp(\pi)$,  then for every $\pi'$ that satisfies $\textrm{supp}(\pi') \subseteq \textrm{supp} (\pi)$, we have
$\overline{u} (\pi) \geq \overline{u} (\pi')$. Hence,
at any credible belief $\pi$ induced by $\sigma^*$, the sender cannot strictly benefit from disclosing any additional information and shifting the receiver's posterior to any $\pi'$ with  $\supp(\pi')\subseteq \supp(\pi)$.\footnote{Using the same constructive proof, one can show that in non-monotone environments, there exists a PEBE in which the sender's payoff is $\overline{V}(\pi_0)$ for every $p \in \Delta (\mathbb{N})$ if every $\pi \in \Delta (\Theta)$ induced by the optimal experiment $\sigma^*$ satisfies a stronger notion of credibility: $\overline{u}(\pi) \geq \underline{u}(\pi')$ for every $\pi'$ that satisfies $\supp(\pi') \subseteq \supp(\pi)$ where $\underline{u}(\pi) \equiv \min_{a \in A^* (\pi) } u^s(a)$. This stronger notion reduces to credibility in monotone environments.} 

Our proof for the second statement uses a lemma that applies beyond monotone environments: If $\sigma^*$ induces at least one non-credible belief, then for each large enough $n \in \mathbb{N}$,  
the sender cannot attain $\overline{V}(\pi_0)$ when
$p$ is close to the uniform distribution on $\{0,1,...,n\}$.
\begin{Lemma}\label{lemma:non-credible}
If $\sigma^*$ induces a non-credible belief, then there exist $\eta>0$, $n \in \mathbb{N}$, and an open set of type distributions that contains the uniform distribution on $\{0,1,...,n\}$ such that for every $p$ in this open set, the sender’s payoff in any PEBE is at most $\overline V(\pi_0)-\eta$.
\end{Lemma}
The proof is in Appendix \ref{secA}. The key step of our proof bounds type $t+1$ sender's payoff from below by computing their payoff from  (i) initially  using the equilibrium strategy of type $t$ and (ii) at every non-credible posterior belief $\pi$ induced by type $t$, conducting an additional experiment that fully reveals the state, and (iii) disclosing the resulting outcome if and only if the corresponding state $\theta$ satisfies $u(\theta)> \overline{u}(\pi)$. This leads to a lower bound on the difference between type $t+1$'s equilibrium payoff and type $t$'s, which is a linear function of the probability with which type $t$'s equilibrium strategy induces non-credible beliefs. Suppose by way of contradiction that there are many types and there exists an equilibrium where the sender's ex ante expected payoff is close to $\overline{V}(\pi_0)$. Then given the uniqueness of the optimal experiment, 
there will be many types inducing non-credible posterior beliefs with probability strictly above zero. Since 
the set of feasible payoffs is bounded, this will lead to a contradiction once the number of types that occur with positive probability is large enough.

Our argument for Lemma \ref{lemma:non-credible} hinges on a \textit{dispersion property} of the receiver's belief $p$ about the sender's type, that it assigns probability bounded above $0$ to sufficiently many types. In fact, in many natural settings, any sufficiently dispersed belief about the sender's type is enough to prevent the sender from attaining their commitment payoff. For each $\theta\in\Theta$, let $\Pi_{\theta}^* \subseteq \Delta (\Theta)$ denote the set of beliefs induced by $\sigma^*$ that have $\theta$ in their support. 
\begin{Theorem}\label{Theorem3}
Suppose there exists a state $\theta\in\Theta$ such that $u(\theta) > \max_{\pi\in\Pi_\theta^*}\overline{u}(\pi)$. Then
there exists a constant $\lambda>0$ such that, for every small enough $\eta>0$, there exists $n\in\mathbb{N}$  such that for any $p \in \Delta (\mathbb{N})$ assigning probability more than  $\lambda\eta$ to at least $n$ distinct types, the sender’s payoff in every PEBE is no more than $\overline{V}(\pi_0)-\eta$.
\end{Theorem}
Theorem \ref{Theorem3} requires the existence of a state
$\theta$ such that whenever the receiver's belief belongs to $\Pi_\theta^*$, the sender strictly prefers to fully reveal $\theta$, i.e., $\overline{u}(\pi) < u(\theta)$.
 In the example of the introduction, this additional requirement is satisfied when the good state occurs with probability between $0$ and $\frac{1}{3}$. This requirement is more demanding than the assumption in Lemma \ref{lemma:non-credible} that  $\sigma^*$ induces at least one non-credible belief, since all beliefs in $\Pi^*_\theta$ are non-credible. Appendix~\ref{secC} provides an example showing that this stronger requirement cannot be weakened to the one in Lemma \ref{lemma:non-credible}.

The proof of Theorem \ref{Theorem3}, in Appendix~\ref{secB},   parallels that of Lemma \ref{lemma:non-credible}:
We again examine a higher-type sender's payoff from imitating the equilibrium strategy 
of a lower type and then conducting an additional experiment 
that fully reveals the state, disclosing the result if and only if it yields 
a strictly higher payoff. 
When $u(\theta)>\max_{\pi\in\Pi_\theta^*}\overline{u}(\pi)$,
the higher type can guarantee a strictly higher payoff than any lower type’s
equilibrium payoff, regardless of the probability that the lower type’s strategy
induces non-credible beliefs.  
The contrapositive—namely, there exists an equilibrium yielding
a sender payoff arbitrarily close to the commitment payoff—must be false 
 once sufficiently many types occur with positive probability, which completes the argument.

Next, we explain why the receiver having a dispersed belief about the sender's type is not redundant. We first present an example where the sender can attain their commitment payoff in some equilibrium even when (i) the optimal experiment induces a non-credible belief and (ii) the receiver faces non-trivial uncertainty about the sender's type. We then establish Theorem \ref{Theorem2}, which shows that when $p$ is concentrated on a single type, there always exists an equilibrium in which the sender’s payoff is close to $\overline{V}(\pi_0)$.

\paragraph{Example:} Suppose $A \equiv \{0,2,3,\frac{7}{2}\}$, $\Theta \equiv \{\theta_0,\theta_1\}$, and $t \in \{0,1\}$. We will use the probability of $\theta_1$ to represent the receiver's belief about the state. 
The sender's payoff equals $a$. The receiver's payoff is such that $a=0$ is optimal when $\pi \in [0,\frac{1}{4}]$, $a=2$ is optimal when $\pi \in [\frac{1}{4},\frac{1}{2}]$, $a=3$ is optimal when $\pi \in [\frac{1}{2},\frac{3}{4}]$, and $a=\frac{7}{2}$ is optimal when $\pi \in [\frac{3}{4},1]$. 
Suppose  $\pi_0 = \frac{3}{8}$ and $t=0$ with probability $\frac{1}{3}$.\footnote{More generally, for every $\pi_0 \in (\frac{1}{4}, \frac{1}{2})$, there is a non-degenerate distribution on $t=\{0,1\}$ and a corresponding equilibrium in which the sender attains $\overline{V}(\pi_0)$.} We construct an equilibrium in which the sender's expected payoff is $\overline{V}(\pi_0) = \frac{5}{2}$, which is attained by inducing posterior beliefs $\frac{1}{4}$ and $\frac{1}{2}$ with equal probability. 

In this equilibrium, the sender chooses an uninformative initial experiment. Type $1$ conducts an additional experiment with outcome $\overline{s}$ occurring with probability $1$ when $\theta=\theta_1$ and probability $\frac{3}{5}$ when $\theta=\theta_0$, and discloses only $\overline{s}$. The receiver's posterior belief equals $\frac{1}{2}$ after observing $\overline{s}$ and $\frac{1}{4}$ after observing no additional signal, matching their naive belief at all off-path information sets.

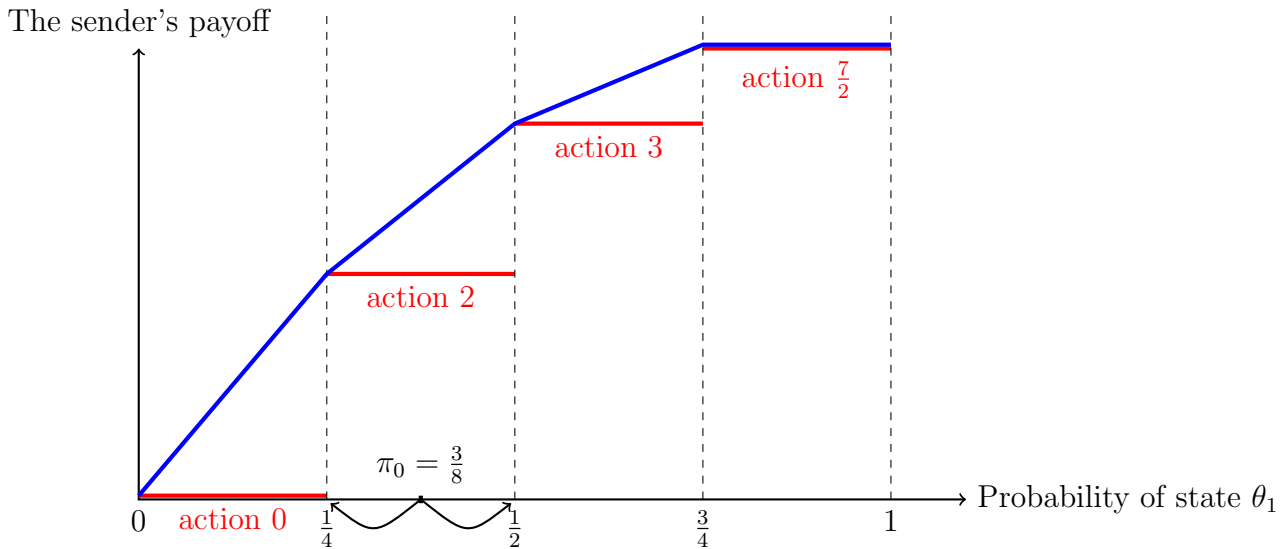
\begin{figure}
\begin{center}
\begin{tikzpicture}[scale=0.5]
\draw[thick, <->] (0,12)node[above]{The sender's payoff}--(0,0)node[below]{$0$}--(22,0)node[right]{Probability of state $\theta_1$};
\draw[ultra thick, red] (0,0.1)--(2.5,0.1)node[below]{action $0$}--(5,0.1);
\draw[ultra thick, red] (5,6)--(7.5,6)node[below]{action $2$}--(10,6);
\draw[ultra thick, red] (10,10)--(12.5,10)node[below]{action $3$}--(15,10);
\draw[ultra thick, red] (15,12)--(17.5,12)node[below]{action $\frac{7}{2}$}--(20,12);
\draw[dashed] (5,0)node[below]{$\frac{1}{4}$}--(5,13);
\draw[dashed] (10,0)node[below]{$\frac{1}{2}$}--(10,13);
\draw[dashed] (15,0)node[below]{$\frac{3}{4}$}--(15,13);
\draw[dashed] (20,0)node[below]{$1$}--(20,13);
\draw[ultra thick, blue] (0,0.1)--(5,6)--(10,10)--(15,12.1)--(20,12.1);
\draw[ultra thick] (7.5,-0.1)--(7.5,0.1)node[above]{$\pi_0=\frac{3}{8}$};
\draw[->, thick]
    (7.5,0) .. controls (6.25,-1) and (6.25,-1) .. (5.1,-0.1);
\draw[->, thick]
    (7.5,0) .. controls (8.75,-1) and (8.75,-1) .. (9.9,-0.1);    
\end{tikzpicture}
\caption{The sender's indirect utility as a function of the receiver's belief $\overline{u}(\pi)$ (\textcolor{red}{red}) and its concave closure (\textcolor{blue}{blue}) in the example.}
\end{center}
\end{figure}

In this example, type $1$ has no incentive to fully reveal state $\theta_1$ via additional experiments. Instead, they prefer to disclose a signal 
that partially reveals the state and induces posterior belief $\frac{1}{2}$. When the receiver only assigns positive probability to types $0$ and $1$, this leads to an equilibrium where the sender obtains their commitment payoff. 
Nevertheless, our construction relies on the assumption that there are only two sender types. As shown in Theorem \ref{Theorem3}, no such equilibrium exists when the sender's type distribution is sufficiently dispersed.

Next, we show that if the
type distribution $p$ assigns probability sufficiently close to $1$ to one type,  there is an equilibrium where the sender's payoff is close to $\overline{V}(\pi_0)$.
\begin{Theorem}\label{Theorem2}
For every $\delta>0$, there exists $\varepsilon>0$ such that for every $p \in \Delta (\mathbb{N})$
where $p(n)>1-\varepsilon$ for some $n \in \mathbb{N}$, there exists a PEBE in which the sender's payoff is more than $\overline{V}(\pi_0)-\delta$.  
\end{Theorem}
Our three theorems together imply that, when payoffs are monotone, whenever the optimal experiment induces some non-credible beliefs, whether the sender can attain their commitment payoff in their optimal equilibrium is governed by the dispersion of the receiver’s belief about their type. If this belief is sufficiently concentrated on a single type (so that the receiver faces little uncertainty about the sender’s capacity), then there exist equilibria in which the sender’s payoff is approximately $\overline{V}(\pi_0)$. By contrast, if the receiver’s belief assigns positive probability to many sender types and each such type occurs with probability bounded above $0$, then the sender’s equilibrium payoff is bounded strictly below their commitment payoff.

Theorems \ref{Theorem3} and Theorem \ref{Theorem2}  are complementary, as their proofs lead to lower and upper bounds on the gap between the sender's commitment payoff and their highest equilibrium payoff. The proof of Theorem \ref{Theorem3} implies that the gap is bounded below by a linear function of the probability of the $n$th most likely type, where $n$ is a function of the sender's benefit from selectively disclosing additional information at non-credible beliefs, the prior belief about the state, and the probability with which the optimal experiment inducing those non-credible beliefs.  
Theorem \ref{Theorem2} provides an upper bound on that gap, which is a linear function of the probability of types other than the most likely one.

The proof of Theorem \ref{Theorem2}, which is in Appendix \ref{secD}, constructs equilibria
that approximately attain the sender's commitment payoff. For some intuition, consider the simple case where type $0$ occurs with  probability zero. In our construction, the sender chooses an uninformative initial experiment and (i) when their realized type is between $1$ and $n-1$, they conduct one  additional experiment that reveals the state and fully disclose the outcome, (ii) when their realized type is at least $n$, they conduct an additional experiment that is close to $\sigma^*$, fully disclose its outcome, as well as the outcome of $n-1$ uninformative experiments, (iii) when their realized type is above $n$, they may conduct more than $n$ additional experiments and selectively disclose their outcomes, and (iv) at every off-path information set, the receiver's posterior belief assigns probability $1$ to the worst state in the support of their naive belief, i.e., the state $\theta$ that minimizes $u(\theta)$. The sender's ex ante expected payoff is close to $\overline{V}(\pi_0)$ given that type $n$ occurs with probability close to $1$. 

This   construction does not work  when type $0$ occurs with positive probability,   because the receiver will then assign probability $1$ to type $0$ when no additional outcome is disclosed, which  can induce types $1$ to $n-1$ not to disclose the  outcome of the fully informative additional experiment.  Our proof adjusts the strategies of types $1$ to $n$ in order to accommodate this. In equilibrium, the sender chooses an experiment close to $\sigma^*$ as their initial experiment. The additional experiments conducted by each type as well as the outcomes  disclosed depend on the relative probabilities of type $0$ and types $1$ to $n-1$. 
When type $0$ is relatively likely compared to types $1$ to $n-1$, type $n$ 
will imitate these lower types  with positive probability, which lowers the lower-types' equilibrium payoffs. 
 When type $0$ is relatively unlikely compared to types $1$ to $n-1$, types $1$ to $n-1$ conduct an additional experiment that fully reveals the state and disclose if and only if their indirect utility under that state is high enough.

\section{Discussion}\label{sec4}
We first present  an example that illustrates how the off-path restrictions of PEBE, which are absent in wPBE, 
force  the sender’s equilibrium payoff below their commitment payoff.  We then examine the sender's worst equilibrium payoff, complementing our earlier analysis of the best equilibrium. Finally, we consider extensions where the sender observes their type before choosing the initial experiment or incurs a cost for each additional experiment.

\paragraph{PEBE \textit{vs} wPBE:} 
Consider the following example: $A = \{0,2,3\}$, $\Theta = \{\theta_0,\theta_1\}$, $t \in \{0,1\}$. The receiver's prior  assigns probability strictly between $0$ and $\frac{1}{3}$ to state $\theta_1$ and $\frac{1}{2}$ to type $t=1$. The sender's payoff equals $a$. 
Let $\pi \in [0,1]$ denote the probability the receiver's posterior belief assigns to state $\theta_1$. 
We assume that it is optimal for the receiver to choose action $0$
when 
$\pi \in [0,\frac{1}{3}]$, action $2$ when $\pi \in [\frac{1}{3},\frac{2}{3}]$, and action $3$ when $\pi \in [\frac{2}{3},1]$.

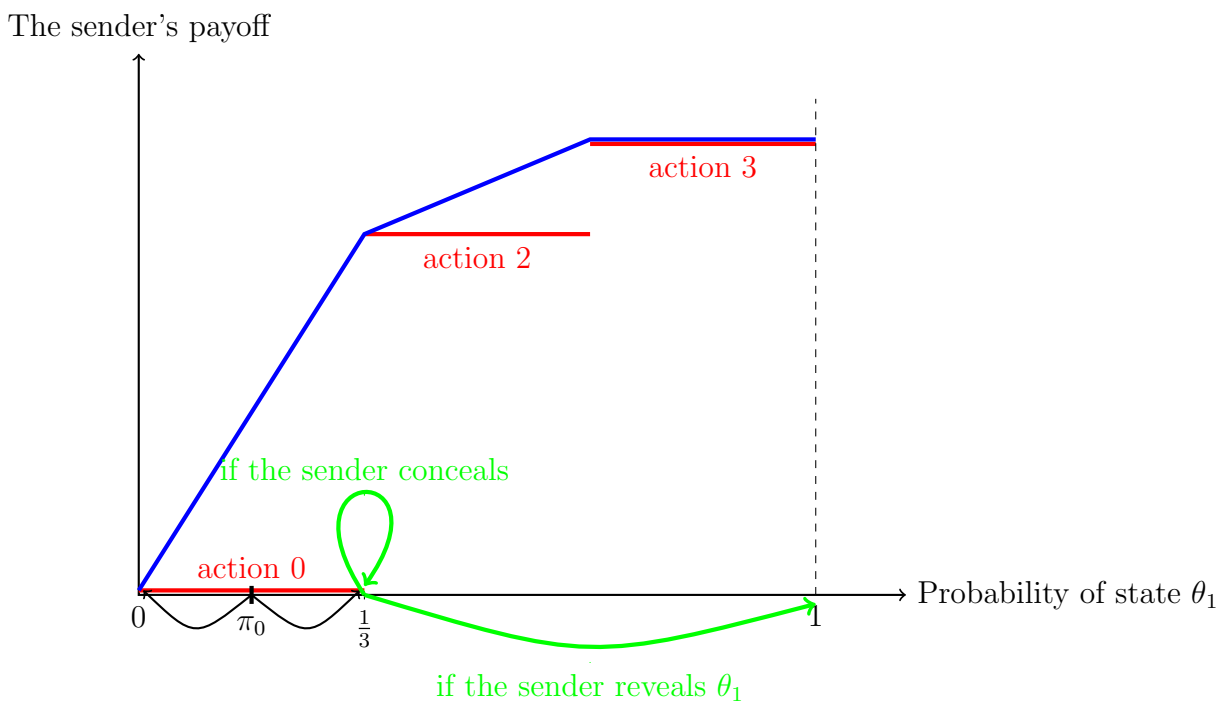
\begin{figure}[h]
\begin{center}
\begin{tikzpicture}[scale=0.6]
\draw[thick, <->] (0,12)node[above]{The sender's payoff}--(0,0)node[below]{$0$}--(17,0)node[right]{Probability of state $\theta_1$};
\draw[ultra thick, red] (0,0.1)--(2.5,0.1)node[above]{action $0$}--(5,0.1);
\draw[ultra thick, red] (5,8)--(7.5,8)node[below]{action $2$}--(10,8);
\draw[ultra thick, red] (10,10)--(12.5,10)node[below]{action $3$}--(15,10);
\draw[dashed] (5,-0.1)node[below]{$\frac{1}{3}$}--(5,0.1);
\draw[dashed] (15,0)node[below]{$1$}--(15,11);
\draw[ultra thick, blue] (0,0.1)--(5,8)--(10,10.1)--(15,10.1);
\draw[ultra thick] (2.5,0.2)--(2.5,-0.2)node[below]{$\pi_0$};
\draw[->, thick] (2.5,0) .. controls (1.25,-1) and (1.25,-1) .. (0.1,0.1);
\draw[->, thick]
    (2.5,0) .. controls (3.75,-1) and (3.75,-1) .. (4.9,0.1);
\draw[->, ultra thick, green]
    (5,0) .. controls (10,-1.5) and (10,-1.5) .. (15,-0.2);
\draw[green] (10,-1.49)--(10,-1.51)node[below]{if the sender reveals $\theta_1$};
\draw[->, ultra thick, green] (5,0) .. controls (3,3) and (7,3) .. (5,0.2);
\draw[green] (5,2.2)--(5,2.3)node[above]{if the sender conceals};
\end{tikzpicture}
\caption{The sender's payoff function $\overline{u}(\pi)$ (\textcolor{red}{red}) and its concave closure (\textcolor{blue}{blue}). The green arrows represent type $1$ sender's profitable deviation after observing the outcome of the initial experiment that leads to naive belief $\frac{1}{3}$.}
\end{center}
\end{figure}

We show that the sender's payoff is strictly below $\overline{V}(\pi_0)$ in every PEBE. Suppose by contradiction that there is a PEBE where the sender's payoff equals $\overline{V}(\pi_0)$. Then the receiver's posterior belief must be $\frac{1}{3}$ conditional on state $\theta_1$ for each sender type. Since both types occur with positive probability, there exists $h \in \mathcal{H}_1$ at which the receiver's posterior belief is $\frac{1}{3}$. But type $1$ then has a strict incentive to conduct a fully informative experiment at $h$ and reveal the outcome if and only if $\theta=\theta_1$, which increases the receiver's action from $2$ to $3$ in state $\theta_1$ while leaving it unchanged in state $\theta_0$. Therefore, there is no PEBE where the sender attains their commitment payoff $\overline{V}(\pi_0)$ (See Figure 3).  

The above argument exploits a key difference between PEBE and wPBE. Under PEBE, whenever the sender discloses evidence revealing state $\theta$ at any information set (on or off the equilibrium path), the receiver's posterior belief must assign probability $1$ to $\theta$. This restriction need not hold under wPBE.\footnote{In addition to the restriction just discussed, PEBE imposes further constraints that wPBE does not. For instance, when the sender's type is bounded above by $t$, the receiver's posterior belief about $\theta$ at any information set where $t$ additional experiment outcomes are revealed must equal the receiver's naive belief.} This is what allows \citet{st2025} to show that the sender obtains their commitment payoff in some wPBE for a class of  payoff environments that includes all cases with transparent motives.

\paragraph{The Sender's Worst Equilibrium Payoff:}
 Our next result shows that, when the sender's type is unbounded, there is an equilibrium in which the sender attains their full-disclosure payoff $\underline{V}(\pi_0)$, the lower bound on their equilibrium payoff established in Proposition~\ref{Prop1}.
\begin{Proposition}\label{Prop5}
If $p \in \Delta (\mathbb{N})$ assigns positive probability to infinitely many types, then there is a PEBE in which the sender fully reveals $\theta$ and conducts no additional experiments on the equilibrium path. 
\end{Proposition}
In settings where the receiver's belief assigns probability close to $1$ to one sender type but also assigns strictly positive but small probability to infinitely many types, Theorem \ref{Theorem2} and Proposition \ref{Prop5} together  imply that there are equilibria in which the sender approximately attains their commitment payoff although there are also equilibria in which the sender obtains their full disclosure payoff. Hence, relative to settings where the receiver knows the sender's type in which case the sender obtains their commitment payoff in all equilibria, the sender's welfare is sensitive to the equilibrium players coordinate on once the sender's type can be arbitrarily large, even when
the receiver is almost certain about their type.  

Our proof of Proposition \ref{Prop5} constructs an equilibrium that is reminiscent of the fully revealing equilibrium in Matthews and Postlewaite (1985), where the receiver forms skeptical beliefs at all information sets where the state is not fully revealed: At every information set $h \in \mathcal{H}$, the receiver's posterior belief will assign probability $1$ to the worst state  in the support of their naive belief $\widehat{\pi}_h$ (i.e., the state under which the receiver's unique best reply minimizes the sender's payoff). When the sender's type distribution $p$ is unbounded, no matter how many additional experiments the sender has revealed, the receiver will always find it plausible that the sender has concealed the outcome of at least one additional experiment and that this outcome reveals the worst state in the support of the receiver's naive belief.

\begin{proof}[Proof of Proposition 2:]
Let $\Theta \equiv \{\theta_1,\theta_2,...,\theta_n\}$. Without loss of generality, we assume that $u(\theta_n) \geq ... \geq u(\theta_1)$. Consider a strategy profile and belief in which the sender chooses an initial experiment that fully reveals the state and conducts no additional experiments on the equilibrium path. For every $k \geq 1$ and  every off-path history $\tilde{h}_{k} \in \mathcal{H}_k$, let $i(\tilde{h}_{k})$ denote the lowest $i$ such that $\theta_i \in \supp(\widehat{\pi}_{\tilde{h}_k})$.  Sender types $t \geq k$ will conduct one fully informative additional experiment at $\tilde{h}_k$ and disclose the outcome unless the state is $\theta_{i(\tilde{h}_k)}$. 
At every off-path information set $\tilde{h}_{k} \in \mathcal{H}_k$, the receiver's posterior belief assigns probability $1$ to state $\theta_{i(\tilde{h}_k)}$, namely, the worst state for the sender in the support of the receiver's naive belief at that information set. One can verify that this is a PEBE in which the state is fully revealed to the receiver on the equilibrium path and the sender's expected payoff is $\underline{V}(\pi_0)$.
\end{proof}

\paragraph{The Sender Observes Their Type Before Conducting the Initial Experiment:}
In the baseline model, the sender observes their type $t$ after choosing the initial experiment. We now show that all our results extend when the sender observes $t$ before choosing the initial experiment, provided we focus on the sender's ex ante expected payoff.

The argument that Proposition \ref{Prop1} extends  has three steps. First, regardless of when the sender observes their type, they can secure their full-disclosure payoff by committing to an initial experiment that fully reveals $\theta$, making subsequent receiver skepticism irrelevant. Second, in equilibrium, the sender's strategy induces a distribution over receiver beliefs with expectation $\pi_0$, so the sender's ex ante expected payoff in any equilibrium cannot exceed $\overline{V}(\pi_0)$. Third, when $p$ is degenerate at type $t$, the sender can approach $\overline{V}(\pi_0)$ by choosing an initial experiment $\sigma^{\varepsilon}$ close to $\sigma^*$ (where the receiver has strict incentives at each induced belief) and conducting $t$ uninformative additional experiments and disclosing all the outcomes.

 Theorem \ref{Theorem3} and
the second statement of Theorem \ref{Theorem1} are  shown via bounding the difference between a higher-type and a lower-type sender's equilibrium payoffs, by considering the higher-type's payoff from a deviation that initially uses the lower type's equilibrium strategy, after which they conduct an additional experiment and selectively disclose the resulting outcome. Our argument does not rely on the assumption that different types conducting the same initial experiment, which implies that our conclusions extend when the sender observes their type $t$ before choosing $\sigma_0$.

Theorem \ref{Theorem2} and
the first statement of Theorem \ref{Theorem1}  extend since there still exists an equilibrium where all types of the sender choose the same initial experiment. Given that every sender type in our baseline model obtains at least  their full disclosure payoff, they will choose to conduct the same initial experiment if the receiver has skeptical beliefs at any information set where the sender deviates to a different initial experiment.

Proposition \ref{Prop5} also extends, since when $p$ assigns positive probability to infinitely many types, there is still an equilibrium in which (i) all types of the sender choose an initial experiment that fully reveals the state and (ii) the receiver's posterior belief at every information set assigns probability $1$ to the state in the support of their naive belief that minimizes the sender's payoff.

\paragraph{Constant Marginal Cost of Additional Experiments:} 
Suppose the sender has a publicly known constant marginal cost $c$ per additional experiment. When $c$ is sufficiently small, our main results extend.  Specifically, the conclusions of  Propositions \ref{Prop1} and \ref{Prop5} hold because the sender can secure their full disclosure payoff by choosing a fully informative initial experiment; and 
when $p$ is unbounded, they can obtain their full disclosure payoff in equilibrium since
choosing a fully informative initial experiment and conducting no additional experiments is an equilibrium as long as the receiver entertains skeptical beliefs. For any degenerate $p$ and $\varepsilon>0$, there exists $\overline{c}>0$ such that when $c<\overline{c}$, the sender's payoff is more than $\overline{V}(\pi_0)-\varepsilon$ in all equilibria. This is because when $c$ is small enough, the sender can secure a payoff that is $\varepsilon$-close to $\overline{V}(\pi_0)$ by choosing
an initial experiment close to $\sigma^*$ where the receiver has strict incentives at each induced belief (which exists by Assumption \ref{Ass1}), conducting $t$ uninformative additional experiments,
and fully disclosing all the outcomes. 
 
The first statement of Theorem \ref{Theorem1} extends to all $c>0$ since our constructive proof does not require the sender to conduct additional experiments on the equilibrium path. The second statement of Theorem \ref{Theorem1} and Theorem \ref{Theorem3} extend when $c$ is small relative to the bounds on the payoff difference between higher-type and lower-type senders. These bounds depend only on the payoff environment $(\pi_0,u^s,u^r)$ and the gap between the sender's equilibrium and commitment payoffs.

 The analog of Theorem \ref{Theorem2} also holds: For every $\delta>0$, 
there exists $\varepsilon>0$ such that for every 
$c< \varepsilon$ and 
 $p \in \Delta (\mathbb{N})$
where $p(n)>1-\varepsilon$ for some $n \in \mathbb{N}$, there exists a PEBE in which the sender's payoff is more than $\overline{V}(\pi_0)-\delta$.

\section{Conclusion}\label{sec5}
We study a novel form of limited commitment, where the sender   commits to an initial experiment  but cannot commit not to conduct additional experiments and  selectively disclose their outcomes. We show that the sender’s ability to attain their commitment payoff depends on two key factors: (i) the precision of the receiver's information about the sender’s ability to conduct additional experiments and 
(ii) whether the sender can benefit from selectively disclosing additional information at beliefs induced by the optimal experiment. When the sender's optimal experiment induces at least one belief under which the sender can benefit from revealing more information (i.e., a non-credible belief) and the receiver is sufficiently uncertain about the sender’s ability to conduct additional experiments, the sender cannot obtain any payoff close to their commitment payoff in any equilibrium. If the environment is monotone and all beliefs induced by $\sigma^*$ are credible, the sender can exactly attain their commitment payoff in some equilibrium; if the receiver is nearly certain of the sender's type, there is an equilibrium where the  sender gets  approximately their commitment payoff.

\clearpage
\appendix
\appendixpage

\section{PEBE in Games with Random Successors}\label{PEBE}
In this Appendix, we extend the definition of PEBE to games where the termination and successor nodes are random and are affected by players' actions.

Consider an extensive-form game with nodes $x \in X$, information sets $h \in H$, actions $a \in A(h)$ available at information set $h$, and signals $s\in S$. We use $h(x)$ to denote the information set that contains $x$. For every $x \in X$ and $a \in A(h(x))$, signal $s \in S$ is realized with probability $q(s \mid x, a)$, leading to successor node $(x,a,s) \in X$. Let $Z$ denote the set of terminal nodes. Given $Y \subseteq X$, we use $Z(Y)$ to denote the subset of terminal nodes that are preceded by some element in $Y$.

A conditional probability system (CPS) on $Z$ is a function $\nu(\cdot \mid \cdot): 2^Z \times 2^Z \setminus \{\emptyset\} \to [0,1]$ such that (i) for all non-empty $\Omega_0 \subseteq Z$, $\nu(\cdot \mid \Omega_0)$ is a probability distribution on $\Omega_0$ and (ii) for all $\Omega_3 \subseteq \Omega_2 \subseteq \Omega_1 \subseteq \Omega$ with $\Omega_2 \neq \emptyset$, we have $\nu(\Omega_3\mid\Omega_2)\, \nu(\Omega_2\mid\Omega_1) =\nu(\Omega_3\mid\Omega_1)$. A CPS $\mu$ on $Z$ induces conditional beliefs $\mu(\cdot \mid \cdot)$: $\mu(A \mid B) \coloneqq \nu(Z(A) \mid Z(B))$ for every $A,B \subseteq X$ such that $Z(A) \subseteq Z(B)$.

Given a profile of strategies $\pi$, let $\pi(a\mid h)$ denote the induced probability of taking action $a \in A(h)$ at information set $h$.
We say a profile of strategies $\pi$ and a CPS $\nu$ on terminal nodes is an extended assessment. 

\begin{Definition}
    Let $\mu$ denote the conditional beliefs associated with $\nu$.
    The extended assessment $(\nu,\pi)$ is a PEBE if 
    \begin{enumerate}
    \item $\pi$ is a best response to conditional beliefs $\mu$.
    
    \item For every information set $h \in H$, action $a \in A(h)$, signal $s \in S$, and node $x \in h$,
    \[
    q(s \mid a, x) \, \pi(a \mid h) = \mu((x,a,s) \mid x);
    \]
    \item For every information set $h \in H$, action $a \in A(h)$, signal $s \in S$, and nodes $x,y \in h$,
    \[
    \frac{\mu((x,a,s) \mid (x,a,s), (y,a,s))}{\mu((y,a,s) \mid (x,a,s), (y,a,s))} 
    =
    \frac{q(s\mid a,x)}{q(s \mid a,y)} \frac{\mu(x \mid x,y)}{\mu(y \mid x,y)}.
    \]
\end{enumerate}
\end{Definition}

In the game we considered in the main text, when the sender is choosing an experiment (either an initial one or an additional one), the action is an experiment $\sigma \in \Sigma$, and the corresponding signal distribution is $\sigma(\cdot \mid x(\theta))$, where $x(\theta)$ is the actual state at node $x$. When the sender is choosing the additional experiments to reveal, or the receiver is choosing an action, the corresponding signal is deterministic and leads to a unique successor. 

\begin{Lemma} \label{Lem: necessary cond PEBE}
    Suppose $(\nu,\pi)$ is a PEBE and $\mu$ is the associated conditional beliefs. For any receiver information set $h$,
    \begin{enumerate}
        \item If $h \in \mathcal{H}_k$, then $\mu(\{t \geq k-1 \} \mid h) = 1$.
        \item $\mu( \{\theta \in \supp(\widehat{\pi}_h)\} \mid h )= 1$.
        \item If $n = \max \supp (p)$ and $h \in \mathcal{H}_{n+1}$, then $\mu(\theta \mid h ) = \widehat{\pi}_h(\theta)$ for every $\theta \in \Theta$.
    \end{enumerate}
\end{Lemma}
\begin{proof}
    Suppose $h \in \mathcal{H}_k$, then for every $x \in h$, the corresponding sender type $t(x) \geq k-1$, so statement 1 holds. 

    Suppose $\theta \not\in \supp(\widehat{\pi}_h)$. For every sender information set $h'$ that is a predecessor of information set $h$, the corresponding sender posterior belief also assigns probability $0$ to $\theta$. Condition 3 in the definition of PEBE then implies that the receiver's posterior at $h$ must assign probability $0$ to $\theta$ as well. 

    If $n = \max \{ \supp ( p) \}$ and $h \in \mathcal{H}_{n+1}$, then for every $x\in h$, the corresponding sender type $t(x) = n$, so the receiver's posterior is $\mu(t=n \mid h) = 1$. This also means the sender must have revealed all the additional experiments they conducted. By condition 3 in the definition of PEBE, the receiver's posterior belief is  the receiver's naive belief $\widehat{\pi}_h$.
\end{proof}

\section{Implications of Assumption \ref{Ass1}} \label{App: Assumptions}
Let $\Pi^* \subseteq \Delta (\Theta)$ denote the set of posterior beliefs induced by $\sigma^*$. 
For any $\pi \in \Delta (\Theta)$ and $\varepsilon>0$, 
let $a^*(\pi)$ denote the sender-preferred receiver best response under belief $\pi$, and let $B(\pi,\varepsilon)$ denote the set of beliefs that has distance less than $\varepsilon$ to $\pi$. The following result shows the main implication of Assumption \ref{Ass1}: there exists a sequence of experiments close to $\sigma^*$, all of which generate strict receiver incentives almost surely.\footnote{For simplicity, in this section we still treat $\sigma^*$ as the unique optimal experiment (Assumption \ref{Ass2}). The main idea of the proof of Lemma \ref{L1} does not rely on this additional assumption.} Lemma 3 of \citet*{lipnowski2024perfect} provides a related result. 

\begin{Lemma}\label{L1}
For any $\varepsilon>0$, there exists an experiment $\sigma_{\varepsilon}^*$ such that for every $\pi \in \Pi^*$, there exists $\pi_{\varepsilon} (\pi) \in B(\pi,\varepsilon)$ that is induced by $\sigma_{\varepsilon}^*$ such that $a^*(\pi)$ is strictly optimal for the receiver when their belief about the state is $\pi_{\varepsilon}(\pi)$, and $\sigma_{\varepsilon}^*$ assigns probability more than $1-\varepsilon$ to beliefs in $\Pi_{\varepsilon}^* \equiv \{\pi_{\varepsilon}(\pi)\}_{\pi\in \Pi^*}$. 
\end{Lemma}

\begin{proof}

\medskip
Since $\Theta$ is finite and $\Delta(\Theta)$ is a simplex in $\mathbb{R}^{|\Theta|}$, any feasible distribution over posteriors with mean $\pi_0$ can be replaced (without changing either the mean or the expected sender's payoff) by another distribution with support size at most $|\Theta|$ via Carath\'eodory’s theorem. Thus we may assume without loss of generality that the optimal experiment $\sigma^*$ has finite support $\Pi^* = \{\pi^1,\dots,\pi^m\}$ with $m\le|\Theta|$, $\sigma^*[\pi^j]>0$, and $\sum_{j=1}^m \sigma^*[\pi^j]\pi^j=\pi_0$.

Let $a_j \equiv a^*(\pi^j)$ for every $j \in \{1,\ldots,m\}$. By Assumption~\ref{Ass1}, for every $j$, there exists $\tilde{\pi}^j \in \Delta(\Theta)$ with $\supp (\tilde{\pi}^j) \subseteq \supp (\pi^j)$ such that $a_j$ is the unique best response $\tilde{\pi}^j$. Therefore, by linearity of expectation, for every $\delta > 0$ and every $j$, there exists $\pi_\varepsilon(\pi^j) \in B(\pi^j, \delta)$ such that $a_j$ is the unique best response at $\pi_\varepsilon(\pi^j)$ and $\supp (\pi_\varepsilon(\pi^j))\subseteq \supp (\pi^j)$.

Next we construct, for each $j$, a probability measure $\tau_j$ on $\Delta(\Theta)$ that (i) has mean $\pi^j$, and (ii) puts (almost) all its mass on $\pi_{\varepsilon}(\pi^j)$. Let
$
\ell \;\equiv\; \min\big\{\pi^j(\theta): j=1,\dots,m,\ \pi^j(\theta)>0\big\},
$
which is strictly positive because $\Pi^*$ is finite and each $\pi^j$ is a probability vector. Fix
$
\delta \;\le\; \frac{\varepsilon \ell}{2} < \varepsilon$
and $
\eta \;\equiv\; \varepsilon\wedge \tfrac12.
$

Define, for each $j$, a vector $\nu^j\in\mathbb{R}^{|\Theta|}$ by
$
\nu^j(\theta) \;\equiv\; 
\frac{\pi^j(\theta) - (1-\eta)\,\pi_{\varepsilon}(\pi^j)(\theta)}{\eta}
\quad\text{for every }\theta\in\Theta$. Observe that 
$$
\sum_{\theta} \nu^j(\theta)
= \frac{\sum_{\theta}\pi^j(\theta) - (1-\eta)\sum_{\theta}\pi_{\varepsilon}(\pi^j)(\theta)}{\eta}
= \frac{1-(1-\eta)}{\eta}
=1
$$ 
To see that $\nu^j$ is non-negative, note that  if  $\pi^j(\theta)=0$, then $\pi_{\varepsilon}(\pi^j)(\theta)=0$ because $\pi_{\varepsilon}(\pi^j)\in F_j$, so $\nu^j(\theta)=0$. If $\pi^j(\theta)>0$, then
$
\pi_{\varepsilon}(\pi^j)(\theta) \le \pi^j(\theta) + \delta,
$
since $\|\pi_{\varepsilon}(\pi^j)-\pi^j\|<\delta$. Hence
\begin{align*}
\pi^j(\theta) - (1-\eta)\,\pi_{\varepsilon}(\pi^j)(\theta)
&\ge \pi^j(\theta) - (1-\eta)\,(\pi^j(\theta)+\delta) \\
&= \eta\,\pi^j(\theta) - (1-\eta)\,\delta.
\end{align*}
Using $\pi^j(\theta)\ge \ell$ and $\delta\le \varepsilon \ell/2$, together with $\eta\ge \varepsilon/2$ (since $\eta=\varepsilon\wedge\frac12$), we have
$
\eta\,\pi^j(\theta) \;\ge\; \eta\,\ell \;\ge\; \frac{\varepsilon \ell}{2}
\;\ge\; \delta,
$
and thus
$
\eta\,\pi^j(\theta) - (1-\eta)\,\delta \;\ge\; \delta - (1-\eta)\,\delta \;\ge\; 0.
$
Therefore $\nu^j(\theta)\ge0$ for all $\theta$, so  $\nu^j\in\Delta(\Theta)$.

Now define $\tau_j$ to be the two-point distribution
$
\tau_j \;\equiv\; (1-\eta)\,\delta_{\pi_{\varepsilon}(\pi^j)} \;+\; \eta\,\delta_{\nu^j}.
$
By construction,
$
\mathbb{E}_{\pi\sim\tau_j}[\pi] 
= (1-\eta)\,\pi_{\varepsilon}(\pi^j) + \eta\,\nu^j
= \pi^j,
$
and
$
\tau_j[\pi_{\varepsilon}(\pi^j)] = 1-\eta \;\ge\; 1-\varepsilon.
$

\medskip

Define a new distribution 
$\sigma_{\varepsilon}^*[\cdot] \equiv \sum_{j=1}^m \sigma^*[\pi^j]\,\tau_j[\cdot]$ over posteriors by first drawing $\pi^j$ according to the original law $\sigma^*$ and then conditional on $\pi^j$, drawing the realized posterior from $\tau_j$.
Because $\mathbb{E}_{\pi\sim\tau_j}[\pi]=\pi^j$, $
\mathbb{E}_{\pi\sim\sigma_{\varepsilon}^*}[\pi]
= \sum_{j=1}^m \sigma^*[\pi^j]\,\mathbb{E}_{\pi\sim\tau_j}[\pi]
= \sum_{j=1}^m \sigma^*[\pi^j]\pi^j
= \mathbb{E}_{\pi\sim\sigma^*}[\pi]
= \pi_0.
$
Hence $\sigma_{\varepsilon}^*\in\Sigma[\pi_0]$ is Bayes-plausible.
 Let $\Pi_{\varepsilon}^* \equiv \{\pi_{\varepsilon}(\pi^j):j=1,\dots,m\}$. For each $j$,
$
\sigma_{\varepsilon}^*[\pi_{\varepsilon}(\pi^j)]
= \sigma^*[\pi^j]\cdot\tau_j[\pi_{\varepsilon}(\pi^j)]
= \sigma^*[\pi^j](1-\eta).
$
Therefore,
$
\sigma_{\varepsilon}^*[\Pi_{\varepsilon}^*]
= \sum_{j=1}^m \sigma^*[\pi^j](1-\eta)
= 1-\eta \;\ge\; 1-\varepsilon.
$
By construction, the receiver’s unique best response at $\pi_{\varepsilon}(\pi^j)$ is $a^*(\pi^j)$ by construction, and $\pi_\varepsilon(\pi^j) \in B(\pi^j,\delta) \subseteq B(\pi^j, \varepsilon)$ since $\delta \leq \frac{\varepsilon \ell}{2} < \varepsilon$. 
Thus we have constructed a Bayes-plausible, finite-support distribution $\sigma_{\varepsilon}^*$ over posteriors with the required properties.
\end{proof}

\section{Proof of Proposition \ref{Prop1}}\label{secAA}
In any equilibrium, the sender's strategy will induce a distribution over the receiver's posterior beliefs about $\theta$ with 
expected value  $\pi_0$. Hence, the sender's equilibrium payoff is no more than their payoff in an auxiliary game where they can choose any distribution over the receiver's posterior beliefs about $\theta$ subject to the constraint that the expectation equals $\pi_0$. 
 By definition, this upper bound is $\overline{V}(\pi_0)$. The sender also has the option to deviate by choosing an initial experiment that fully reveals  $\theta$ and conducting no additional experiments regardless of the realized $\theta$. Under such a deviation, for every $\theta \in \Theta$, the sender obtains payoff $u(\theta)$ when the realized state is $\theta$. 
 Hence, the sender obtains an expected payoff $\underline{V}(\pi_0)$ from this deviation, which implies that
their equilibrium payoff must be at least  $\underline{V}(\pi_0)$.

In order to show that when $p$ is degenerate the sender's payoff is $\overline{V}(\pi_0)$ in all equilibria,  let  
$\underline{u}(\pi) \equiv \min_{a \in A^*(\pi)} u^s(a)$ for every $\pi \in \Delta (\Theta)$ and let
\begin{equation*}
\overline{V}^{\sigma} \equiv \sum_{\pi \in \Delta (\Theta)} \sigma[\pi] \overline{u}(\pi) \quad \textrm{and} \quad
\underline{V}^{\sigma} \equiv \sum_{\pi \in \Delta (\Theta)} \sigma[\pi] \underline{u}(\pi),
\end{equation*}
which are the sender's expected payoffs when the receiver's information structure is $\sigma$, who breaks ties in favor of and against the sender, respectively. Since $\Theta$ and $A$ are finite, the set $\arg \max_{\sigma \in \Sigma}\overline{V}^{\sigma}$ is non-empty. 
Recall that $\sigma^*$ is the optimal experiment. 
By definition, $\overline{V}(\pi_0)=\overline{V}^{\sigma^*}$.  
By Lemma \ref{L1}, for every $\varepsilon>0$,
there exists an experiment $\sigma^{\varepsilon}$  such that $\underline{V}^{\sigma^{\varepsilon}} > \overline{V}^{\sigma^*}-\varepsilon$. 

Suppose $p$ is degenerate and assigns probability $1$ to type $t^*$. Fix any equilibrium and $\varepsilon>0$. Consider the sender's payoff under the following deviation: Conduct an initial experiment $\sigma^{\varepsilon}$ that satisfies $\underline{V}^{\sigma^{\varepsilon}} > \overline{V}^{\sigma^*}-\varepsilon=\overline{V}(\pi_0)-\varepsilon$. 
 The sender then conducts $t^*$ additional experiments, all of which are uninformative,
regardless of the realized outcome of the initial experiment. They then disclose the outcomes of all these additional experiments to the receiver. After observing $t^*$ additional outcomes, the receiver's posterior belief assigns probability $1$ to the sender disclosing all the outcomes, which implies that their posterior belief about $\theta$ must coincide with their naive belief. 
This implies that the sender's payoff is at least $\overline{V}(\pi_0)-\varepsilon$ under such a deviation. Since $\varepsilon>0$ can be arbitrary, the sender's payoff in every equilibrium is at least $\overline{V}(\pi_0)$. 

\section{Proof of Lemma \ref{lemma:non-credible}}\label{secA}
Since every degenerate belief is credible and $\sigma^*$ is assumed to be the unique optimal experiment, the hypothesis that $\sigma^*$ induces a non-credible belief implies that $\overline{V}(\pi_0)> \underline{V}(\pi_0)$. 
Let $\underline{\pi} \in \Delta (\Theta)$ denote a non-credible belief induced by $\sigma^*$. By definition, there exists $\theta \in \textrm{supp}(\underline{\pi})$ such that $u(\theta) > \overline{u}(\underline{\pi})$. Recall that $B(\pi,d)$ denotes the open ball with radius $d>0$ centered at $\pi$. For every $\sigma \in \Sigma[\pi_0]$, denote  the conditional distribution over beliefs given state $\theta$ by $\sigma^{\theta}$, where
$\sigma^\theta[B] \equiv \int_{B} \frac{\pi(\theta)}{\pi_0(\theta)} d\sigma[\pi]$ for every Borel measurable set $B \subseteq \Delta(\Theta)$.

\begin{Lemma} \label{lem: approx opt dist lower bound}
Suppose $r$ is such that $\pi(\theta) > \underline{\pi}(\theta)/2$ for every $\pi \in B(\underline{\pi},r)$. There exist $\eta, \varepsilon > 0$ such that for every $\sigma \in \Sigma[\pi_0]$, if $\mathbb{E}_{\pi\sim \sigma}\overline u(\pi) > \overline{V}(\pi_0) - \eta$, then $\sigma^\theta[B(\underline{\pi},r)] \geq \varepsilon$.
\end{Lemma}
\begin{proof}[Proof of Lemma \ref{lem: approx opt dist lower bound}]
        Let $m \equiv  \sigma^*[B(\underline{\pi},r)] > 0$ and let $\mathcal{C}:= \{\sigma \in \Sigma[\pi_0]: \sigma[B(\underline{\pi},r)] \leq m/2 \}$. Since $B(\underline{\pi},r)$ is an open set, the function that maps $\sigma$ to $\sigma[B(\underline{\pi},r)]$ is lower semi-continuous by the Portmanteau Theorem, so the lower sublevel set $\{\sigma \in \Delta(\Delta(\Theta)): \sigma[B(\underline{\pi},r)] \leq m/2\}$ is closed. Since $\Theta$ is finite, $\Sigma[\pi_0]$ is compact, and hence $\mathcal{C}$ is a closed subset of a compact set, which is also compact.

        By definition, $\sigma^* \notin \mathcal{C}$. Since $\sigma^*$ is the unique optimal experiment (Assumption \ref{Ass2}),  $\sup_{\sigma\in\mathcal{C}} \mathbb{E}_{\pi\sim\sigma}(\overline u(\pi)) < \overline{V}(\pi_0)$. Let
        $\eta \equiv \min \{ \overline{V}(\pi_0)- \sup_{\sigma\in\mathcal{C}} \mathbb{E}_{\sigma}(\overline{u}), m/2\} > 0$ and $\varepsilon \equiv \frac{\underline{\pi}(\theta)}{2\pi_0(\theta)} \eta > 0$. It follows that for any $\sigma \in \Sigma[\pi_0]$ with $\mathbb{E}_{\pi\sim\sigma}(\overline u(\pi)) > \overline{V}(\pi_0) - \eta$, $\sigma \notin \mathcal{C}$ and hence $\sigma[B(\underline{\pi},r)] > m/2 \geq \eta$. By definition, 
        \[
        \sigma^\theta[B(\underline{\pi},r)] = \int_{B(\underline{\pi},r)} \frac{\pi(\theta)}{\pi_0(\theta)} d \sigma[\pi] \geq \frac{\underline{\pi}(\theta)}{2\pi_0(\theta)} \sigma[B(\underline{\pi},r)] \geq \frac{\underline{\pi}(\theta)}{2\pi_0(\theta)} \eta = \varepsilon,
        \]
        where the first inequality follows from the assumption that $\pi(\theta) > \underline{\pi}(\theta)/2$ for every $\pi \in B(\underline{\pi},r)$. 
\end{proof}
 Let $\overline{v}$ and $\underline{v}$ denote the highest and lowest feasible payoff for the sender, respectively. Since $\overline{u}(\pi)$ is an upper semi-continuous function with respect to $\pi$, there exists $r>0$ such that for every $\pi \in B(\underline{\pi},r)$, we have $u(\theta)> \overline{u}(\pi) $ and $\pi(\theta) > \underline{\pi}(\theta)/2 \geq 0$. By Lemma \ref{lem: approx opt dist lower bound}, there exists $\eta,\varepsilon > 0$ such that for every $\sigma \in \Sigma[\pi_0]$, if $\mathbb{E}_{\pi\sim \sigma}\overline u(\pi) > \overline{V}(\pi_0) - \eta$, then $\sigma^\theta[B(\underline{\pi},r)] \geq \varepsilon$.
       Fix this selection of $r, \eta, \varepsilon$ and choose $N \in \mathbb{N}$ with
    \[
    N > \frac{3(\overline{v} - \underline{v})}{2\varepsilon \pi_0(\theta) (u(\theta) - \overline{u}(\underline{\pi}))} + \frac{2}{\varepsilon}.
    \]
    Endow the set of type distributions $\Delta(\mathbb{N})$ with the product topology. For any $p \in \Delta (\mathbb{N})$ that belongs to 
    \[
    U \equiv \Big\{p \in \Delta(\mathbb{N}): |p(t) - \tfrac{1}{N}| < \tfrac{1}{2N}\, \forall t=0,\ldots,N-1, \textrm{ and } p(\{t \geq N\}) < \tfrac{1}{2N} \Big\},
    \]
    we show that the sender's expected payoff in any equilibrium is no more than $\overline{V}(\pi_0) - \eta$.

Suppose by way of contradiction that there exists an equilibrium in which the sender's payoff is strictly greater than $\overline{V}(\pi_0) - \eta$. Let $\widetilde{\sigma}$ denote the  equilibrium distribution of the receiver's posterior beliefs. 

For any $t \in \{0,1,...,N\}$, 
pick any pure strategy that is used by type $t$ with positive probability in equilibrium and let $\mathcal{H}_t^*$ denote the set of information sets such that for every $h \in \mathcal{H}_t^*$, under that pure strategy used by type $t$,
 (i) $h$ will be reached with positive probability, and that type $t$ will conduct no additional experiments afterwards and (ii) the equilibrium probabilities with which they disclose the outcomes of the additional experiments at $h$ will induce a receiver posterior belief that belongs to $B(\underline{\pi},r)$.

For every $t \geq 0$, consider type $t+1$'s payoff when they deviate to the following strategy: Use the same strategy as type $t$ except for histories that belong to $\mathcal{H}_{t}^*$. At every $h \in \mathcal{H}_{t}^*$, conduct one additional experiment that fully reveals the state, and discloses the outcome of that experiment if and only if the state is $\theta$. Let $V_{t+1}$ denote the equilibrium payoff of type $t+1$ and $\widetilde{\sigma}_{t+1}^{\theta}[\cdot]$ denote the induced  distribution of receiver posterior beliefs given type $t+1$'s equilibrium strategy when the state is  $\theta$. Type $t+1$'s payoff from the above deviation is at least
\begin{equation*}
V_{t}+ \widetilde{\sigma}_{t}^\theta[B(\underline{\pi},r)] \pi_0(\theta) (u(\theta)-  \overline{u}(\underline \pi)),
\end{equation*}
since they will induce the same receiver-action as one of the pure strategies type $t$
uses with positive probability in equilibrium, 
except that when the state is $\theta$ and type $t$ reaches information sets that belong to $\mathcal{H}_{t}^*$, type $t+1$'s deviation will induce the receiver to take their optimal action for state $\theta$. Hence, for every $t \geq 0$, we have
\begin{equation*}
V_{t+1}-V_{t} \geq \widetilde{\sigma}_{t}^\theta[B(\underline{\pi},r)] \pi_0(\theta) \Big(u(\theta)-  \overline{u}(\underline{\pi}) \Big).
\end{equation*}
It follows that
\[
    V_{N-1} - V_0 \geq  \left( \sum_{t=0}^{N-2} \widetilde{\sigma}_t^\theta[B(\underline{\pi},r)] \right) \pi_0(\theta) \Big(u(\theta)-  \overline{u}(\underline{\pi})\Big).
\]
Since $|p(t) - \frac{1}{N}| \leq \frac{1}{2N}$ for every $t = 0,\ldots, N-1$, we have  
\[
(\tfrac{1}{N} + \tfrac{1}{2N}) \sum_{t=0}^{N-2}\tilde{\sigma}_t^{\theta}[B(\underline{\pi},r)] \geq \sum_{t=0}^{N-2} p(t)\widetilde{\sigma}_t^{\theta}[B(\underline{\pi},r)] = \widetilde{\sigma}^{\theta}[B(\underline{\pi},r)] - \sum_{t=N-1}^{\infty} p(t)\widetilde{\sigma}_t^{\theta}[B(\underline{\pi},r)]
\]
By Lemma \ref{lem: approx opt dist lower bound}, we have 
\[
\widetilde{\sigma}^{\theta}[B(\underline{\pi},r)] - \sum_{t=N-1}^{\infty} p(t)\widetilde{\sigma}_t^{\theta}[B(\underline{\pi},r)] \geq \varepsilon - p(N-1) - p(\{t\geq N\}) \geq \varepsilon - \tfrac{2}{N}, 
\]
where the last inequality follows from $|p(N-1) - \frac{1}{N}| \leq \frac{1}{2N}$ and $p(\{t\geq N\}) \leq \tfrac{1}{2N}$.  Therefore,
\[
    V_{N-1} - V_0  \geq (\tfrac{2}{3}N\varepsilon - \tfrac{4}{3}) \pi_0(\theta) (u(\theta)-  \overline{u}(\underline \pi)).
\]
This together with the construction of $N$ leads to a contradiction.

\section{Proof of Theorem \ref{Theorem3}}\label{secB}
Since $u(\theta) > \max_{\pi \in \Pi_{\theta}^*} \overline{u}(\pi)$ and the best reply correspondence is upper-hemi-continuous, there exists $\varepsilon>0$ such that for every $\pi'$ such that  the Hausdorff distance between $\pi'$ and $\Pi_{\theta}^*$ is less than $\varepsilon$, we have $\overline{u}(\pi') \leq \max_{\pi \in \Pi_{\theta}^*} \overline{u}(\pi)$. 
Since the optimal experiment is unique and the sender's payoff is bounded, we know that there exists $\lambda_0 >0$ such that for every $\eta>0$ and every equilibrium where the sender's payoff is more than $\overline{V}(\pi_0)-\eta$, we have that conditional on every $\theta \in \Theta$, 
 with probability at least $1-\lambda_0 \eta$,
the receiver will reach information sets where their posterior beliefs 
will have Hausdorff distance less than $\eta$ to $\Pi_{\theta}^*$. Let $\overline{v}$ and $\underline{v}$ denote the sender's highest and lowest feasible payoff, respectively.

Let $\lambda \equiv 2 \lambda_0$ and
\begin{equation*}
n \equiv \Big\lceil \frac{2(\overline{v}-\underline{v})}{\displaystyle \pi_0(\theta) \Big\{ u(\theta) - \max_{\pi \in \Pi_{\theta}^*} \overline{u}(\pi) \Big\} } \Big\rceil+1,
\end{equation*}
and pick any $\eta$ such that 
\begin{equation*}
\eta < \min \left\{\varepsilon, \frac{1}{\lambda n}\right\}.
\end{equation*}
Suppose by way of contradiction that $p$ assigns probability more than $\lambda \cdot \eta$ to at least $n$ types yet there exists an equilibrium where the sender's payoff is more than $\overline{V}(\pi_0)-\eta$. By construction, the receiver's posterior belief will have Hausdorff distance less than $\eta$ to set $\Pi_{\theta}^*$ with probability more than $1-\lambda \eta/2$ conditional on $\theta$. Let $t_1<...<t_n$ denote $n$ types that occur with probability at least $\lambda \eta$. 
Let $\Pr(\cdot \mid \theta,t)$ denote the induced probability measure given the equilibrium strategy of type $t$ when the state is $\theta$. We use $D$ to denote the Hausdorff distance between the induced receiver posterior and $\Pi^*_\theta$, which is a random variable. Taking expectation across types, we have $\mathbb{E}_{t\sim p}\left[\Pr(D \ge \eta\mid\theta,t)\right]<\lambda\eta/2$. Markov's inequality then implies that $p( \left\{t \in \Delta(\mathbb{N}):\Pr(D \ge \eta\mid\theta,t)\ge 1/2\right\})<\lambda\eta$. 
Because each selected type $t_i$ has probability at least $\lambda\eta$, none of them can be in that set, so for every $t_i$ we have $\Pr(D<\eta\mid\theta,t_i)>1/2$. For every $i \in \{1,2,...,n\}$, 
pick any pure strategy that is used with positive probability in equilibrium by type $t_i$ and let $\mathcal{H}_i^*$ denote the set of sender histories such that for every $h \in \mathcal{H}_i^*$, under that pure strategy used by type $t_i$,
 (i) they will reach $h$ with positive probability after which they will conduct no additional experiments and (ii) the equilibrium probabilities with which they disclose the outcomes of the additional experiments
 at $h$ will lead to a receiver posterior belief that has Hausdorff distance less than $\eta$ to $\Pi_{\theta}^*$.

For every $i \geq 2$, consider type $t_i$'s payoff when they deviate to the following strategy: Use the same strategy as type $t_{i-1}$ except for histories that belong to $\mathcal{H}_{i-1}^*$. At every $h \in \mathcal{H}_{i-1}^*$, conduct one additional experiment that fully reveals the state, and discloses the outcome of that experiment if and only if the state is $\theta$. Let $V_i$ denote the equilibrium payoff of type $t_i$. Type $t_i$'s payoff from the above deviation is at least
\begin{equation*}
V_{i-1}+ \frac{1}{2} \pi_0(\theta) \Big\{
u(\theta)- \max_{\pi \in \Pi_{\theta}^*} \overline{u}(\pi)
\Big\},
\end{equation*}
since they will induce the same receiver action as one of the pure strategies type $t_{i-1}$
uses with positive probability in equilibrium, 
except that when the state is $\theta$ and type $t_{i-1}$ reaches history $\mathcal{H}_{i-1}^*$, they will induce receiver action $a^*(\theta)$. Hence, for every $i \geq 2$, we have
\begin{equation*}
V_i-V_{i-1} \geq \frac{1}{2} \pi_0(\theta) \Big\{
u(\theta)- \max_{\pi \in \Pi_{\theta}^*} \overline{u}(\pi)
\Big\}.
\end{equation*}
This together with the construction of $n$ leads to a contradiction.

\section{The Requirement in Theorem \ref{Theorem3} is not Redundant}\label{secC}
We use an example to show that the requirement in Theorem \ref{Theorem3} cannot be 
weakened to that in Lemma \ref{lemma:non-credible}. 
Suppose $A \equiv \{0,2,3\}$, $\Theta \equiv \{\theta_0,\theta_1\}$, and $t \in \{0,1\}$. 
The receiver's prior belief assigns probability $\pi_0$ to state $\theta_1$.
The sender's payoff equals $a$. 
The receiver's optimal action is $0$ if $\pi \in [0,1/3]$, is $2$ if $\pi \in [1/3,2/3]$, and is $3$ if $\pi \in [2/3,1]$, and the prior state distribution  is $\pi_0 = 1/2$, so the optimal experiment induces beliefs $1/3$ and $2/3$. Belief $1/3$ is not credible,  but the sender's payoff under state $\theta_1$ is the same as their payoff when they induce belief $2/3$.  
Suppose the prior type distribution $p \in \Delta(\mathbb{N})$ is given by  $p(0)=1/3$ and $p(t) = \frac{1}{3}\frac{1}{2^{t-1}}$ for all $t \in \mathbb{N} \setminus \{0\}$. 

We will show that this example has an equilibrium where (i) the receiver's posterior belief is $1/3$ at every information set where no additional outcome is disclosed and is $2/3$ at every on-path information set where one additional outcome is disclosed, (ii) the sender conducts an uninformative initial experiment, (iii) if the sender's realized type $t$ is at least $1$, the sender will conduct $t$ additional experiments, all of which have two signal realizations $\{\underline{s},\overline{s}\}$ with $\sigma(\underline{s}|\theta_0)=2/3$ and $\sigma(\overline{s}|\theta_1)=1$, and will only disclose the first realized outcome that is $\overline{s}$, and (iv) the receiver's posterior belief is $0$ at every off-path information set with a non-degenerate naive belief. 

To see that such an equilibrium exists,  we first verify that the receiver's posterior beliefs on the equilibrium path are consistent with Bayes rule. After observing an additional outcome $(\sigma, \overline{s})$, the receiver's posterior belief is 
\[
    \frac{\pi_0(1-p(0))}{\sum_{t=1}^{\infty}p(t) (1-(1-\pi_0)(\frac{2}{3})^t)} 
    = \frac{\frac{1}{3}}{\frac{1}{3}\sum_{t=1}^\infty\frac{1}{2^{t-1}}(1-\frac{1}{2}(\frac{2}{3})^t)} 
    = \frac{2}{3}.
\]
After observing no additional outcome, the receiver's posterior belief is 
\[
    \frac{\pi_0 p(0)}{p(0)+\sum_{t=1}^\infty p(t)(1-\pi_0)(\frac{2}{3})^t} = \frac{\frac{1}{6}}{\frac{1}{3} + \frac{1}{3}\sum_{t=1}^\infty \frac{1}{2^t}(\frac{2}{3})^t} = \frac{1}{3}.
\]
Next, we check that no sender type has an incentive to deviate. For type $t$, the on-path payoff can be calculated recursively: first conduct the experiment $\sigma$, reveal the outcome $(\sigma, \overline{s})$ if $\overline{s}$ realizes; otherwise, continue experimenting until a signal $\overline{s}$ realizes in the remaining $t-1$ experiments. $\overline{s}$ realizes with probability $2/3$, leading to payoff $\overline{u}(2/3)$, and otherwise the sender gets a continuation payoff $V_{t-1} = (1-(2/3)^{t-1})\overline{u}(2/3) + (2/3)^{t-1} \overline{u}(1/3)$, which is the expected payoff from repeating $\sigma$ for $t-1$ times given $\theta = \theta_0$. Hence, type $t$'s on-path payoff is $\frac{2}{3} \overline{u}(2/3) + \frac{1}{3} V_{t-1}$. The most profitable deviation for type $t$ is to first conduct a fully informative additional experiment, reveal the outcome if $\theta = \theta_1$, and otherwise repeat the experiment $\sigma$ for $t-1$ times and only disclose the first realized outcome that is $\overline s$. We may also calculate the payoff recursively, which is $\frac{1}{2}\overline{u}(1) + \frac{1}{2}V_{t-1}$. The on-path payoff dominates the best deviation as $V_{t-1} < 3 = \overline{u}(2/3)$.

\section{Proof of Theorem \ref{Theorem2}}\label{secD}
When the optimal experiment $\sigma^*$, which is unique by Assumption \ref{Ass2}, fully reveals $\theta$, we have $\overline{V}(\pi_0)=\underline{V}(\pi_0)$, in which case 
Proposition \ref{Prop1} will imply that there exists an equilibrium that attains payoff $\overline{V}(\pi_0)$. In what follows, we focus on the case where $\sigma^*$ does not fully reveal $\theta$. 
By Lemma \ref{L1},  for every $\delta>0$ there exists $\varepsilon>0$ and an experiment $\sigma^*_\varepsilon$ such that for every $\pi \in \Pi^*$, there exists $\pi_{\varepsilon} (\pi) \in B(\pi,\varepsilon)$ that is induced by $\sigma_{\varepsilon}^*$ such that $a^*(\pi)$ is strictly optimal for the receiver when their belief about the state is $\pi_{\varepsilon}(\pi)$, and $\sigma_{\varepsilon}^*$ assigns probability more than $1-\varepsilon$ to beliefs in $\Pi_{\varepsilon}^* = \{\pi_{\varepsilon}(\pi)\}_{\pi\in \Pi^*}$. 
The sender's expected payoff from committing to any such   $\sigma_{\varepsilon}^*$ is thus at least $\overline{V}(\pi_0)-\delta$.

Fix a type distribution $p\in \Delta(\mathbb{N})$ that assigns probability $p(n) > 1-\varepsilon$ to some $n \in \mathbb{N}$. We show that there exists a PEBE where the sender's payoff is more than $\overline{V}(\pi_0) - \delta$. If $\supp (p)$ is bounded above, let $N \equiv \max \supp (p)$ denote the highest possible sender type. We focus on the case where $n < N$ or $\supp (p)$ is not bounded above (there's no such $N$), because if $n = N$, using a similar argument as Proposition \ref{Prop1}, the type $n$ sender can secure a payoff that is close to their commitment payoff. It follows that the sender's ex ante payoff is also close to the commitment payoff since $p(n) \approx 1$.

If $n < N$ or $\supp (p)$ is not bounded above, we construct an equilibrium as follows:

\begin{itemize}
    \item The sender conducts an initial experiment $\sigma_{\varepsilon}^*$. The receiver's naive belief after observing the outcome of this initial experiment is called their \textit{interim belief}. 
    \item Following any credible interim belief in $\Pi^*_\varepsilon$, all types of the sender except the highest type $N$ (if it exists) conduct no additional experiments in equilibrium.
    \item Following any non-credible interim belief in $\Pi^*_\varepsilon$, types $1$ to $n-1$ conduct a fully informative additional experiment, and reveal some of the states while concealing others. Type $n$ conducts $n$ uninformative experiments and discloses all outcomes. They may also imitate lower types, depending on the relative probabilities of type 0 and types $1$ to $n-1$; we discuss two cases separately in subsection \ref{subsec: noncred interim}. Types strictly greater than $n$ except $N$ conduct $n$ uninformative experiments, fully disclose all outcomes, and then conduct another fully informative experiment, concealing unfavorable state reports. 
    \item Interim beliefs not in $\Pi^*_\varepsilon$ occur with probability less than $\varepsilon$ under $\sigma^*_\varepsilon$ and hence have a negligible impact on the sender's ex ante expected payoff.
    \item The receiver's on-path posteriors are consistent with Bayes rule. At any off-path information set $h \not\in \mathcal{H}_{N+1}$, if the sender reveals any additional signal that is informative and induces a non-degenerate naive belief $\pi$, then the receiver's posterior belief assigns probability $1$ to some $\theta \in \supp (\pi)$ that satisfies
    \begin{equation*}
    u(\theta) \leq u(\theta') \textrm{ for every } \theta' \in \supp(\pi).
    \end{equation*}
\end{itemize}

In this equilibrium, following any interim belief $\pi \in \Pi^*_\varepsilon$, type $n$'s on-path payoff is close to $\overline{u}(\pi)$ and hence the sender's ex ante expected payoff is close to $\overline{V}(\pi_0)$ as $p(n) \approx 1$. Given the receiver's off-path posterior, it is without loss of generality to focus on subgames where the sender does not deviate to other initial experiments. In the paragraph below, we analyze subgames following a credible interim belief in $\Pi^*_\varepsilon$. In subsection \ref{subsec: noncred interim}, we consider subgames following non-credible interim beliefs in $\Pi^*_\varepsilon$ and describe the details of the strategy profiles.

Following any credible interim belief $\pi \in \Pi_{\varepsilon}^*$, all types of the sender except the highest type $N$ (if it exists) conduct no additional experiments in equilibrium. Doing so is optimal for all types of the sender given the receiver's off-path belief and the facts that (i) $\pi$ is credible and (ii) type $N$ has negligible probability, so the receiver's posterior when observing no additional experiments is close to $\pi$ and also credible.

\subsection{Non-credible Interim Beliefs} \label{subsec: noncred interim}

Consider the subgame following 
a non-credible interim belief 
 $\pi \in \Pi_{\varepsilon}^*$.
Let $\supp(\pi) \equiv \{\theta_1,...,\theta_k\}$,   
$\pi_j \equiv \pi(\theta_j)$, and
$u_j \equiv u(\theta_j)$ for $j\in \{1,\ldots,k\}$. 
Without loss of generality, we assume that $u_1 \geq u_2 \geq ... \geq u_k$. Let $u^* \equiv u^s(a^*(\pi))$. 
Since $\pi \in \Pi_{\varepsilon}^*$ is non-credible and full disclosure is suboptimal at $\pi$, we know that 
$u_1 > u^* > u_k$. 

Let $u_i^*$ be defined via the following equation:
\begin{equation}\label{indifferencetypen}
u^*= \sum_{j=1}^i \pi_j u_j + \Big(1-\sum_{j=1}^i \pi_j \Big) u_i^*.
\end{equation}
Intuitively, $u_i^*$ is the sender's payoff from no disclosure so that he is indifferent between receiving $u^*$ and disclosing the state if and only if it belongs to $\{\theta_1,...,\theta_i\}$. By definition, $u_k^*=+\infty$. 
The sender's incentive to reveal states with index below $i$ and to conceal states with index above $i$ requires that $u_i \geq u_i^* \geq u_{i+1}$. Lemma \ref{L2} shows that this incentive constraint is equivalent to $u_i^*$ reaching a local minimum at $i$. 
\begin{Lemma}\label{L2}
$u_i \geq u_i^* \geq u_{i+1}$ if and only if $\min\{u_{i+1}^*,u_{i-1}^*\} \geq u_i^*$. 
\end{Lemma}
\begin{proof}
Applying equation (\ref{indifferencetypen}) to $i$ and $i+1$, we have:
\begin{equation}\label{iandi+1}
\sum_{j=1}^i \pi_j u_j + \Big(1-\sum_{j=1}^i \pi_j \Big) u_i^* = \sum_{j=1}^{i+1} \pi_j u_j + \Big(1-\sum_{j=1}^{i+1} \pi_j \Big) u_{i+1}^*.
\end{equation}
Equation (\ref{iandi+1}) is equivalent to
\begin{equation*}
 \Big(1-\sum_{j=1}^{i+1} \pi_j + \pi_{i+1}\Big) u_i^*=
\pi_{i+1} u_{i+1} + \Big(1-\sum_{j=1}^{i+1} \pi_j \Big) u_{i+1}^*,
\end{equation*}
which can be rewritten as
\begin{equation}\label{inequalityA}
\Big(1-\sum_{j=1}^{i+1} \pi_j \Big) (u_{i}^*-u_{i+1}^*) = \pi_{i+1} (u_{i+1}-u_i^*).
\end{equation}
This implies that $u_{i+1}^* \geq u_i^*$ if and only if $u_i^* \geq u_{i+1}$. 
Substituting $\sum_{j=1}^{i+1}\pi_j=\sum_{j=1}^{i}\pi_j+\pi_{i+1}$ and regrouping terms, we obtain the following from
(\ref{iandi+1}):
\begin{equation*}
\Big(1-\sum_{j=1}^{i} \pi_j\Big) u_i^*=
\pi_{i+1} u_{i+1} + \Big(1-\sum_{j=1}^{i} \pi_j -\pi_{i+1}\Big) u_{i+1}^*,
\end{equation*}
from which we obtain that
\begin{equation}\label{inequalityB}
\Big(1-\sum_{j=1}^{i} \pi_j \Big) (u_i^*-u_{i+1}^*) = \pi_{i+1} (u_{i+1}-u_{i+1}^*).
\end{equation}
This implies that $u_{i+1} \geq u_{i+1}^*$ if and only if $u_i^* \geq u_{i+1}^*$. An analogous argument shows $u_{i-1}^{*}\geq u_{i}^{*}\;\Longleftrightarrow\;u_{i}\ge u_{i}^{*}.$ 
\end{proof}
\begin{Lemma}\label{L3}
There exists $i \in \{1,2,...,k\}$ such that $u_i \geq u_i^* \geq u_{i+1}$. 
\end{Lemma}
\begin{proof}
Consider two cases. First, if $u_1^* \geq u_2^*$, then the fact that $u_k^*=+\infty$ implies that there exists $i$ such that  $\min\{u_{i+1}^*,u_{i-1}^*\} \geq u_i^*$. Lemma \ref{L2} implies that $u_i \geq u_i^* \geq u_{i+1}$ for such $i$. Second, if $u_2^* > u_1^*$, then inequality (\ref{inequalityA}) implies that $u_1^* \geq u_2$. To show that $i=1$ satisfies our requirement, we only need to show that $u_1 \geq u_1^*$. Suppose by way of contradiction that $u_1< u_1^*$, then given that $u_1\pi_1+ (1-\pi_1) u_1^* = u^*$, we have $u_1 < u^* < u_1^*$. This contradicts an implication that $\pi$ is non-credible which is $u_1 > u^*$. 
\end{proof}
Define strategy $s^i$ to be the sender pure strategy (in the subgame following interim belief $\pi$) where the sender conducts only one fully informative additional experiment and discloses the resulting outcome if and only if $\theta \in \{\theta_1,...,\theta_i\}$. 
Recall that type $n$ occurs with probability close to $1$. Define strategy $s^{nd}$ to be a sender pure strategy where
the sender conducts $n$ uninformative experiments and discloses all outcomes. 
By definition, strategy $s^i$ is only feasible for types of at least $1$ and strategy $s^{nd}$ is feasible only for types of at least $n$. 
Let $\mathbb{G} \subseteq \Delta (\Theta) \times \mathbb{R}$ be such that $(\pi,u) \in \mathbb{G}$ if and only if there exists $\alpha \in \Delta (A)$ such that $\alpha$ is the receiver's best reply at posterior belief $\pi$ and $\mathbb{E}_{a\sim \alpha}[u^s(a)]=u$. By definition, $\mathbb{G}$ is closed and connected.

Fix any $i$ that satisfies $u_i \geq u_i^* \geq u_{i+1}$, which exists by Lemma \ref{L3}. Let $\pi_*^i$ denote the receiver's posterior belief 
that assigns probability $\frac{\pi_j}{\sum_{\tau=i+1}^k \pi_{\tau}}$ to state $\theta_j$ for $j \in \{i+1,...,k\}$ and assigns zero probability to other states. 
When $\varepsilon$ is sufficiently small, the fact that $\pi \in \Pi_{\varepsilon}^*$
implies that $\pi$ is close to an interim belief $\pi'$ under the sender's optimal experiment $\sigma^*$ and that $a^*(\pi)=a^*(\pi')$.
Because $\sigma^*$ is  the optimal experiment, we know that  
for every $a \in \arg\max_{a' \in A} u^r(\pi_*^i,a')$, 
\begin{equation}\label{IC-satisfied}
u^* > \sum_{j=1}^i \pi_j u_j + \Big(1-\sum_{j=1}^i \pi_j \Big) u^s(a). 
\end{equation}
At interim belief $\pi$,
let $\pi^i$ denote the receiver's posterior belief about $\theta$ at a history that is induced by type $0$ for every $\theta \in \{\theta_1,...,\theta_k\}$, by types $1$ to $n-1$ if and only if $\theta \in \{\theta_{i+1},...,\theta_k\}$, and is never induced by types $n$ and above. Formally, let $p(t)$ denote the probability that the receiver's prior belief assigns to type $t$. Belief $\pi^i$ is then defined as:
\begin{equation*}
\pi^i(\theta_j) \equiv \frac{p(0) \pi_j}{p(0) + \sum_{t=1}^{n-1} p(t) \cdot  \sum_{s=i+1}^k \pi_s} \textrm{ for every } j \leq i
\end{equation*}
and
\begin{equation*}
\pi^i(\theta_j) \equiv \frac{\sum_{t=0}^{n-1} p(t) \pi_j}{p(0) + \sum_{t=1}^{n-1} p(t) \cdot  \sum_{s=i+1}^k \pi_s } \textrm{ for every } j \geq i+1.
\end{equation*}
We consider two cases, depending on the receiver's best reply at $\pi^i$.
\paragraph{Case A:} Suppose there exists $a \in \arg\max_{a' \in A} u^r(\pi^i,a')$ such that 
\begin{equation}\label{IC-violated}
u^* \leq \sum_{j=1}^i \pi_j u_j + \Big(1-\sum_{j=1}^i \pi_j \Big) u^s(a).
\end{equation}
Intuitively, this is the case where the probability of type $0$ is high relative to the probability of types $1$ to $n-1$. 
The definition of $u_i^*$ in (\ref{indifferencetypen}) implies that $u^s(a) \geq  u_i^*$. 

In equilibrium following interim belief $\pi$, sender  types
strictly greater than $n$ (except the highest type $N$ if it exists) conduct $n$ uninformative experiments, fully disclose all outcomes, and then conduct another fully informative experiment and disclose its outcome if and only if the realized state $\theta=\theta_j$ satisfies $u_j>u^*$.  Sender type $N$ may use a different strategy, but it is different from $s^i$ or $s^{nd}$. Sender types $1$ to $n-1$ use strategy $s^i$. 

Sender type $n$ mixes between strategy $s^{nd}$ and strategy $s^i$. If type $n$ uses strategy $s^{nd}$ for sure, the receiver's posterior when observing no additional experiment is $\pi^i$. If type $n$ uses strategy $s^i$ for sure, as it occurs with probability close to $1$, the receiver's posterior $\widetilde{\pi}^i_*$ when observing no additional experiment is close to $\pi^i_*$, so there exists $a \in \arg\max_{a'\in A} u^r(\widetilde{\pi}^i_*, a')$ satisfying \eqref{IC-satisfied}. Since graph $\mathbb{G}$ is closed and connected, type $n$ can mix between strategies $s^i$ and $s^{nd}$ such that
under the receiver's posterior belief when no additional outcome is disclosed, they have a (potentially mixed) best reply $\alpha \in \Delta (A)$ such that $\mathbb{E}_{a\sim \alpha}[u^s(a)]=u_i^*$, which makes type $n$ indifferent between $s^i$ and $s^{nd}$. 

\paragraph{Case B:} Suppose
next that
$u^* > \sum_{j=1}^i \pi_j u_j + \Big(1-\sum_{j=1}^i \pi_j \Big) u^s(a)$
for every  $a \in \arg\max_{a' \in A} u^r(\pi^i,a')$.
Intuitively, this is the case where the probability of type $0$ is low relative to the probability of types $1$ to $n-1$. By definition, $u^s(a)<u_i^*$. 
Since the construction of $i$ requires that $u_i^* \in [u_{i+1},u_i]$, 
we know that $u^s(a) < u_i$. 

On the equilibrium path following interim belief $\pi$, sender types strictly greater than $n$  (except the highest type $N$ if it exists)  will conduct $n$ uninformative experiments, fully disclose all outcomes, and then conduct another fully informative experiment and disclose its outcome if and only if the realized state is $\theta=\theta_j$ that satisfies $u_j>u^*$. Sender type $N$ may use a different strategy, but it is different from $s^j$ or $s^{nd}$. Sender types $1$ to $n-1$ mix between strategy $s^j$ and strategy $s^{j+1}$ for some $j \geq i$. Sender type $n$ uses strategy $s^{nd}$ for sure.

We show that there exist $j \geq i$ as well as a mixing probability for types $1$ to $n-1$ between strategy $s^j$ and strategy $s^{j+1}$ such that under the receiver's posterior belief after observing no additional outcome being disclosed, there exists $\alpha \in \Delta (A)$ that is optimal for the receiver such that if the receiver takes action $\alpha$ when there is no disclosure, then the following two incentive constraints are satisfied. First,
\begin{itemize}
\item When both strategies $s^j$ and $s^{j+1}$ are played with positive probability, the sender's payoff from strategies $s^j$ and $s^{j+1}$ are weakly higher than those from strategy $s^{\tau}$ for every $\tau \in \{1,2,...,k\}$.
\item When only strategy $s^j$ is played with positive probability, the sender's payoff from strategy $s^j$ is weakly higher than that from strategy $s^{\tau}$ for every $\tau \in \{1,2,...,k\}$.
\end{itemize}
This ensures that sender types $1$ to $n-1$ have no incentive to deviate to alternative strategies. 
Second, the sender's expected payoffs from strategies $s^j$ and $s^{j+1}$ are no more than $s^{nd}$. This ensures that sender type $n$ has no incentive to imitate types $1$ to $n-1$.

A \textit{disclosure rule} is a mapping from $\{\theta_1,...,\theta_k\}$ to the probabilities with which each state is disclosed. A \textit{monotone disclosure rule} is one in which for every $\theta_{j}$ that is disclosed with positive probability, every state $\theta_{\tau}$ with $\tau <j$ is disclosed with probability $1$. By definition, each monotone disclosure rule can be pinned down by its probability of disclosure, denoted by $q$. The strategies for types $1$ to $n-1$ described above map to monotone disclosure rules.

We show there is a monotone disclosure rule under which the first incentive constraint is satisfied. 
Let $V(q) \subseteq \mathbb{R}$ denote the set of sender's expected payoff from disclosing no additional signal when the receiver believes that sender types $1$ to $n-1$  use monotone disclosure rule with disclosure probability $q$. By definition, $u^* \in V(1)$ since only type $0$ discloses no additional signal when types $1$ to $n-1$'s disclosure probability is $1$.
Let $\mathbb{V} \subseteq [0,1] \times \mathbb{R}$ be such that $(q,v) \in \mathbb{V}$ if and only if $v \in V(q)$. 
Let $\mathbb{U} \subseteq [0,1] \times \mathbb{R}$ be such that (i) for every $q$ that can be written as $q=\sum_{j=1}^{\tau} \pi_{j}$ for some integer $\tau$, we have $(q,v) \in \mathbb{U}$ if and only if $v \in [u_{\tau+1},u_{\tau}]$ and (ii) for every $q$ that satisfies $\sum_{j=1}^{\tau} \pi_j< q < \sum_{j=1}^{\tau+1} \pi_j$,  we have $(q,v) \in \mathbb{U}$ if and only if $v=u_{\tau+1}$. 
Both $\mathbb{V}$ and $\mathbb{U}$
are continuous, and moreover, graph $\mathbb{U}$ is decreasing in $q$ with $(1,u_k) \in \mathbb{U}$. 
The hypothesis for this case requires that when $q=q^* \equiv \sum_{j=1}^i \pi_j$, every payoff in $V(q^*)$ is less than $u_i$. 
Since $u_k < u^*$, there exists an  intersection $(\hat{q},\hat{v})$ between $\mathbb{V}$ and $\mathbb{U}$ with $\hat{q} \geq q^*$. 
By definition, any monotone disclosure rule with disclosure probability $q$ and the receiver's action that leads to non-disclosure payoff $v$ will satisfy the first incentive constraint.

In the last step, we verify that the intersection $(\hat{q},\hat{v})$ also satisfies the second incentive constraint. If the intersection is at $\hat{q}=q^*$, then the second incentive constraint is trivially satisfied since $u_i^* \in [u_{i+1}, u_i]$ and $u^* > \sum_{j=1}^i \pi_j u_j + \Big(1-\sum_{j=1}^i \pi_j \Big) u^s(a)$
for every  $a \in \arg\max_{a' \in A} u^r(\pi^i,a')$. In what follows, we focus on the case where the intersection is at $\hat{q}>q^*$. 
For every $\tau \in \{1,2,...,k\}$ and $\beta \in [0,1]$, let $u^*(\tau,\beta)$ be defined as
\begin{equation*}
u^*(\tau,\beta)
= \frac{u^*-\sum_{j=1}^{\tau}\pi_j u_j-\beta\,\pi_{\tau+1}u_{\tau+1}}
{1-\sum_{j=1}^{\tau}\pi_j-\beta\,\pi_{\tau+1}},
\end{equation*}
or equivalently 
\begin{equation}\label{indifferencetypenmixed}
u^*= \sum_{j=1}^{\tau} \pi_j u_j + \beta \pi_{\tau+1} u_{\tau+1}+
\Big(1-\sum_{j=1}^{\tau} \pi_j  - \beta \pi_{\tau+1} \Big) u^*(\tau,\beta).
\end{equation}

Intuitively, suppose the sender mixes between strategy $s^{\tau}$ and strategy $s^{\tau+1}$ with probabilities $1-\beta$ and $\beta$, respectively. Then 
if the sender's payoff from disclosing no additional signal is $u^*(\tau,\beta)$, they are  indifferent between 
strategy $s^{nd}$ and the mixed strategy mentioned earlier in this paragraph. 
We use this auxiliary value $u^*(\tau,\beta)$ to show that in the equilibrium we constructed, the sender's actual payoff from not disclosing an  additional signal is less than the corresponding $u^*(\tau,\beta)$ where $\tau,\beta$ are pinned down by $\sum_{j=1}^\tau \pi_j + \beta \pi_{\tau+1} = \hat{q}$, and hence the second incentive constraint is satisfied.

By definition, $u^*(\tau,0)=u_{\tau}^*$ as defined in (\ref{indifferencetypen}) and $u^*(\tau, 1) = u^*(\tau+1,0)$. 
Recall that $i$ is chosen such that
$u_i^* \equiv u^*(i,0) \geq u_{i+1}$. 
In the proof of Lemma \ref{L2}, we show that $u^*_{\tau+1} \geq u^*_\tau$ if and only if $u_\tau^* \geq u_{\tau+1}$. An induction argument shows that $u^*_{\tau+1} \geq u^*_\tau$ for every $\tau \geq i$.\footnote{The base case $\tau = i$ is true. Suppose $u^*_{\tau+1} \geq u^*_\tau \geq u_{\tau + 1}$, then $u^*_{\tau+1} \geq u^*_\tau \geq u_{\tau+1} > u_{\tau + 2}$, which implies that $u^*_{\tau+2} \geq u^*_{\tau+1}$.} By definition, $u^*(\tau,\beta)$ is increasing in $\beta$ for every $\tau \geq i$ since
\begin{align*}
    \frac{\partial u^*(\tau,\beta)}{\partial \beta} =& \pi_{\tau+1}\frac{u^* - u_{\tau+1} + \sum_{j=1}^\tau \pi_j(u_{\tau+1} - u_j)}{(1-\sum_{j=1}^\tau \pi_j -\beta \pi_{\tau+1})^2} \\
    =& \pi_{\tau+1}\frac{(1-\sum_{j=1}^i \pi_j) u_i^* - \sum_{j=i+1}^\tau \pi_j u_j - \sum_{j=\tau+1}^k \pi_j u_{\tau+1}}{(1-\sum_{j=1}^\tau \pi_j -\beta \pi_{\tau+1})^2} \geq 0,
\end{align*}
where the inequality follows from $u_i^* \geq u_{i+1}$. For monotone disclosure rules, it implies that the function $u^*(\tau,\beta)$, is strictly increasing in the disclosure probability when it is between $q^* = \sum_{j=1}^i \pi_j$ and the intersection $\hat{q}$.

Since $\mathbb{U}$ is decreasing in $q$, $v \leq u_{i+1}$ for every $(q,v) \in \mathbb{U}$ with $q> q^*$. Therefore, suppose the intersection $\hat{q} = \sum_{j=1}^{\tau} \pi_{\tau} + \beta \pi_{\tau+1}$ for some $\tau \geq i$ and $(\hat{q},\hat{v}) \in \mathbb{U}$, then $\hat{v} \leq u_{i+1} \leq u^*(i,0) \leq u^*(\tau,\beta)$, where the last inequality follows from monotonicity of $u^*(\tau,\beta)$. This implies that any intersection between $\mathbb{U}$ and $\mathbb{V}$ where $\hat{q}>q^*$ will satisfy the second incentive constraint.

\end{spacing}
\newpage
\bibliography{bibliography/refs}

\end{document}